\documentclass[journal]{IEEEtran}
\usepackage{graphicx}  
\usepackage{amssymb,amsfonts,amstext,acronym}
\usepackage[absolute,overlay]{textpos}
\usepackage{times}
\usepackage{tikz}
\usepackage{verbatim}
\usepackage{xspace}                       
\usepackage{algorithm,algpseudocode}
\usepackage{siunitx}                        
\usepackage{setspace}                              
\usepackage{amsmath,amssymb}
\usepackage{siunitx}                    
\usepackage{subfigure}                      

\graphicspath{{figures/}}
\newcommand{\fig}[1]{Fig.~\ref{#1}}
\newcommand{\tab}[1]{Table~\ref{#1}}

\newcommand{\sect}[1]{Section~\ref{#1}}

\newcommand{\tx}{\textnormal}
\DeclareMathOperator{\sat}{sat}

\usepackage{amsthm}
\theoremstyle{plain}
\newtheorem{theorem}{Theorem}
\newtheorem{lemma}{Lemma}
\newtheorem{proposition}{Proposition}

\newtheorem{remark}{Remark}

\newlength \figwidth
\begin{document}


\title{Design and comparative analyses of optimal feedback controllers for hybrid electric vehicles}


\author{Maryam~Razi,~Nikolce~Murgovski,~Tomas~McKelvey~and Torsten~Wik
\IEEEcompsocitemizethanks{\IEEEcompsocthanksitem The authors are with the Department of Electrical Engineering, Chalmers University of Technology, Gothenburg, Sweden (E-mail: razim@chalmers.se).%
\IEEEcompsocthanksitem This work has been funded by the Area of Advance Energy, Chalmers University.}
}%

\maketitle

\begin{abstract}                         
    This paper presents an adaptive equivalent consumption minimization strategy (ECMS) and a linear quadratic tracking (LQT) method for optimal power-split control of combustion engine and electric machine in a hybrid electric vehicle (HEV). The objective is to deliver demanded torque and minimize fuel consumption and usage of service brakes, subject to constraints on actuator limits and battery state of charge (SOC). We derive a function for calculating maximum deliverable torque that is as close as possible to demanded torque and propose modeling SOC constraints by tangent or logarithm functions that provide an interior point to both ECMS and LQT. We show that the resulting objective functions are convex and we provide analytic solutions for their second order approximation about a given reference. We also consider robustness of the controllers to measurement noise using a simple model of noise. Simulation results of the two controllers are compared and their effectiveness is discussed.
\end{abstract}

\section{INTRODUCTION} \label{sec:intro}
    Hybrid electric vehicles (HEVs) can improve fuel economy and reduce pollutant emissions by propelling the vehicle with multiple actuators. HEVs possess an internal combustion engine (ICE) and at least one electric machine (EM) \cite{guzzella13}, and require a controller to optimally split demanded power between the ICE and EM, while minimizing fuel consumption and simultaneously satisfying constraints on battery state of charge (SOC), drivability, emissions, etc.
    
    The power-split control for energy management of HEVs can be divided into rule-based and optimization-based methods \cite{wirasingha11}. HEV applications with rule-based methods, often designed by system
    criteria or with fuzzy logic, have been assessed by \cite{vagg16, kim11, bianchi10, pei17,wu16, solanomartinez12, schouten03}. These methods are generally easy to implement and are computationally suitable for real time applications, but they cannot guarantee optimal solution and the process of finding suitable system criteria may be cumbersome \cite{rezaei18}.
    
    Optimization-based controllers generally require a model of the HEV powertrain and are able to guarantee local or global optimum. These methods typically employ predictive information about a future driving cycle, described by road gradient and velocity profiles as a function of distance \cite{li19, zhou19, johannesson15b}. Optimal power-split decision along the driving cycle is often ensured via calculus of variations or Pontryagin's minimum principle (PMP) \cite{murgovski13TVT}. A well-known PMP inspired method for optimal control of HEVs is the equivalent consumption minimization strategy (ECMS), where electric energy is approximated to equivalent amount of fuel by the use of the battery costate, also known as the equivalent factor \cite{lin19, zhou17}. Then, the main challenge in ECMS is transcribed to finding an optimal trajectory for the equivalent factor that satisfies battery SOC constraints. If SOC constraints are imposed only at the initial and final instance of the driving cycle, the solution of the PMP methods generally leads to a constant equivalent factor which is obtained numerically by solving a two-point boundary value problem with iterative approaches \cite{delprat04}. However, a constant equivalent factor cannot prevent violation of SOC constraints at intermediate instances. Several approaches have been proposed to mitigate this problem, but the most common are the horizon-split technique \cite{keulen14, uebel18} and approaches where the equivalent factor is adapted by a SOC feedback.
    
    ECMS approaches where the equivalent factor is adapted with a SOC feedback have been proposed in \cite{tianheng15, pisu07, enang17}. The authors in \cite{tianheng15} propose a proportional-integral (PI) controller to calculate the equivalent factor that tracks a given SOC reference. A nonlinear PI controller is proposed in \cite{pisu07}, where the proportional part is modeled by a cubic penalty function of the normalized SOC value. A robust proportional ECMS controller is designed in \cite{enang17}, where the proportional term is modeled with a tangent function.
    
    Linear quadratic controllers have also been proposed for optimal control of HEVs. 
    These controllers are well-known and highly useful in the optimal control theory.
    They are employed in many applications due to their simple structures and explicit and stable solutions.
    A linear quadratic optimal control has been applied in \cite{xia16} for splitting power and tracking a constant reference SOC and two feedback coefficients have been obtained for electric mode and hybrid mode operation of HEV. 
    A linear quadratic tracking (LQT) method has been used to optimize route and speed for a given origin-destination pair with the given expected trip time by \cite{zhou19asme}.
    Linear quadratic regulator (LQR) has been employed to control vehicle speed by \cite{saeed16}, to improve drivability and reduce excessive clutch wear during the transition from electric to hybrid driving by \cite{du17} and to correct the battery power in hybrid mode, promote charge sustainability and make the final SOC at the end of optimization horizon be close to its target by \cite{dafeng18}. 
    These controllers have been designed without considering SOC constraints in the optimization problem and fuel consumption in the LQR's objective function. 

    One of the most important requirements of power-split control in an HEV is to deliver the torque demanded by supervisory controller and/or driver. 
    While solving a power-split optimization problem in HEVs with ECMS and LQR/LQT methods, 
    it is often assumed in the reviewed literature that the demanded torque is fully delivered by the ICE and/or the EM. However, it is possible in practice that demanded torque cannot be delivered because of constraints violation,  e.g. SOC bounds.
    In this case, the delivered torque should be maximized.
    In \cite{becerra16}, an LQT controller has been designed where the difference between demanded and delivered power and tracking errors of fuel mass and the battery SOC have been parts of its objective function. The resulting solution is a compromise between power tracking requirements, fuel consumption minimization and SOC tracking. Furthermore, the SOC bounds have not been considered in the problem and references for battery and fuel mass have been constant. In practical applications the trajectory references for the battery SOC and fuel mass can change according to a  driven cycle.
    
    
    This paper provides several contributions to overcome the shortcomings of previous approaches. By considering that demanded torque may not be deliverable, as the first aim of the HEV optimal control 
    we derive the deliverable torque limit that is as close as possible to the demanded torque. 
    This also allows decreasing the number of control inputs in the optimization problem to one, thus significantly reducing the computational effort for computing the optimal control. 
    Furthermore, the proposed input reduction is made such that usage of service brakes is avoided when demanded torque is negative. 
    Second, we explicitly impose constraints on the battery SOC and we design and investigate adaptive ECMS and LQT controllers. 
    In order to derive computationally efficient analytic solutions, the SOC constraints are modeled by nonlinear interior point penalty functions. 
    For the adaptive ECMS control, we compare the performance between the commonly used tangent function with novel logarithm penalty functions. The latter are also used in a novel LQT controller that trades off fuel consumption, time-varying SOC reference tracking and proximity to SOC constraints. We prove that the resulting objective functions in both ECMS and LQT are convex and we obtain a sub-optimal solution for the optimization problems analytically by using second order approximation of their objective functions.
    Moreover, the robustness of controllers to measurement noise is investigated by considering a simple  noise model.
    
    The rest of the paper is organized as follows. 
    \sect{sec:problemdescription} states the optimal control problem for a parallel HEV. 
    In \sect{sec:convex_problem} a compact convex problem formulation is given.  
    In \sect{sec:controller}, two adaptive power-split control methods are presented. Simulation results are given in \sect{sec:simulation}. Finally, conclusions are drawn in \sect{sec:conclusion}.
\section{Problem description}
\label{sec:problemdescription}
    In this section, we state the power-split optimization problem for an HEV powertrain in a parallel configuration.
\subsection{Multi-layer control architecture}
    The optimal energy management of an HEV has the objective of minimizing fuel consumption by incorporating predictive information of the road and traffic ahead \cite{guzzella13}. The emerging control problem is dynamic, mixed-integer and  non-convex  program  that  is  intractable  for  real-time  computation. The typical approach for obtaining a computationally tractable solution is splitting the problem into a finite number of sub-problems, organized into several control layers. 
    Candidates for computationally efficient multi-layer control architectures of HEVs and conventional vehicles have been proposed in \cite{johannesson15a,cerofolini14,murgovski16CEP,uebel19}. For example, the architecture proposed in \cite{johannesson15a} includes several hierarchical levels, all with a similar objective of minimizing fuel consumption, but concentrating on different system states and physical constraints and featuring different control horizons, sampling intervals and update frequencies. The control layers can generally be categorized into two groups; the higher levels include supervisory controllers that generate reference set points, and the lower layers consist of local controllers that track the set points and deliver control signals for the vehicular actuators (see \fig{fig:controller}). In this way it is possible to design designated local controllers that focus on a specific control functionality, such as the power-split control problem studied in this paper with the aim of minimizing fuel consumption and keeping the battery SOC within its bounds.
\begin{figure}[t]
    \centering
    \vspace{-15pt}
    \centerline{\includegraphics[width=1.05\linewidth]{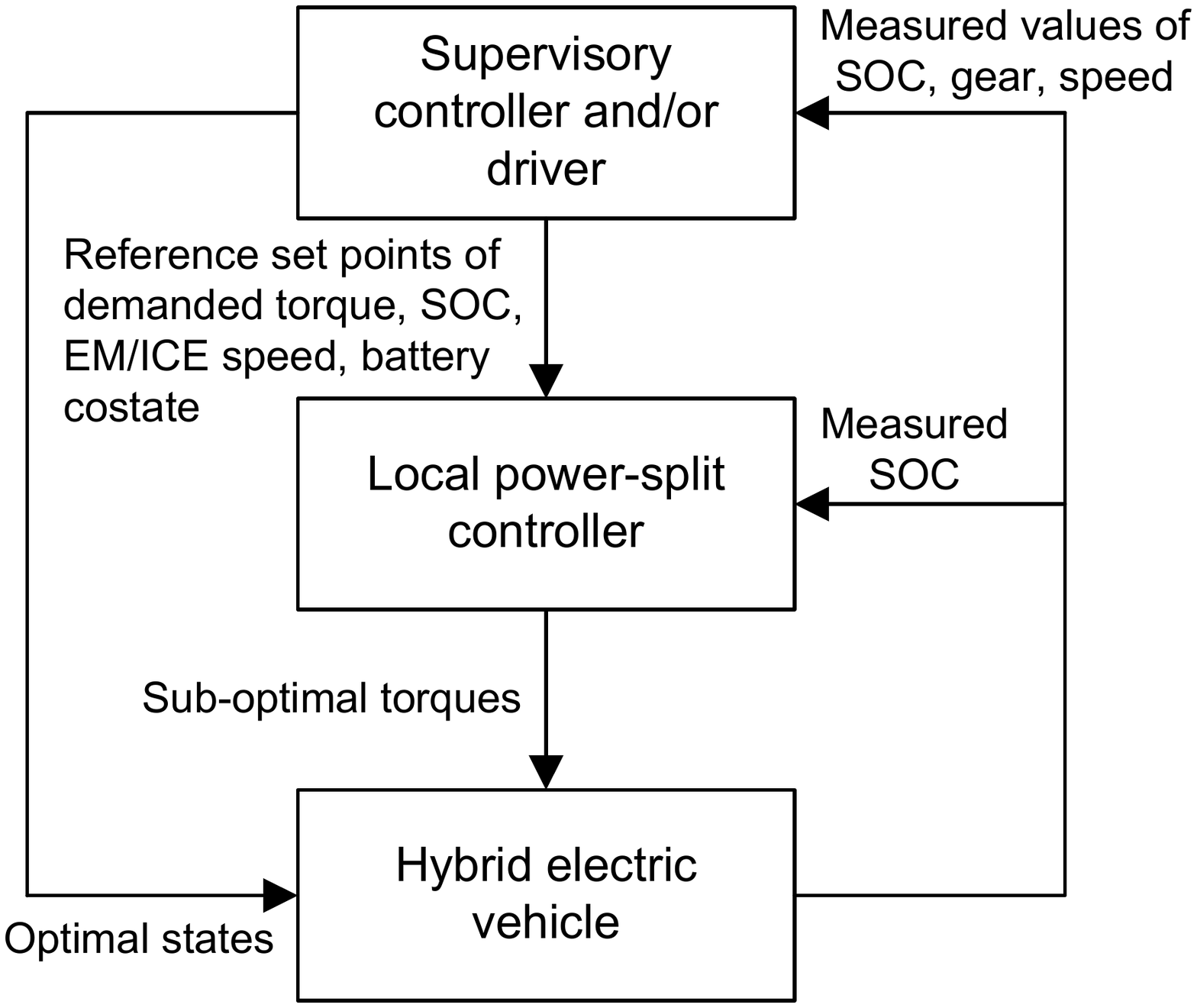}}
    \vspace{-20pt}
    \caption{Control architecture for an HEV, consisting of supervisory controllers, which may also include the driver's demands, and a local power-split controller that tracks the reference trajectories requested by the supervisor.}
    \label{fig:controller}
\end{figure}
    Typical set points generated by the supervisory controllers, and/or a vehicle's driver, may include reference trajectories for the vehicle speed, transmission gear, engine on/off, clutches, battery SOC, etc. 
    These could be predictive trajectories planned by a supervisor over a look ahead horizon, or instantaneous requests when prediction is not available or the requests are given by the driver.
    
    The objective of the local power-split controller is to track references set by the supervisor, but above all, to make sure that control and state bounds are satisfied. The latter is especially relevant when look-ahead information is uncertain, limited or unavailable \cite{pisu07,buerger16}, as tracking such references may lead to infeasibility or poor controller performance. 
    In this paper, we design local power-split controllers that robustly satisfy control and state bounds and perform optimally even when prediction is not available, or reference requests are infeasible.
\subsection{Model of a parallel HEV powertrain}
    We consider an HEV powertrain in a parallel configuration, as illustrated in \fig{fig:parallelhev}. The powertrain includes an ICE and an EM that both are able to deliver power to the wheels through a gearbox. When mechanical power is negative, the EM may operate as a generator, thus recuperating electric energy that can be stored in a battery for a later use. The powertrain includes two clutches that can disengage both actuators or only the ICE. Since this paper studies optimal power-split decisions, only a hybrid mode of operation where both clutches are engaged, is considered.
\begin{figure}[t]
    \centering
    \includegraphics[width=0.68\linewidth]{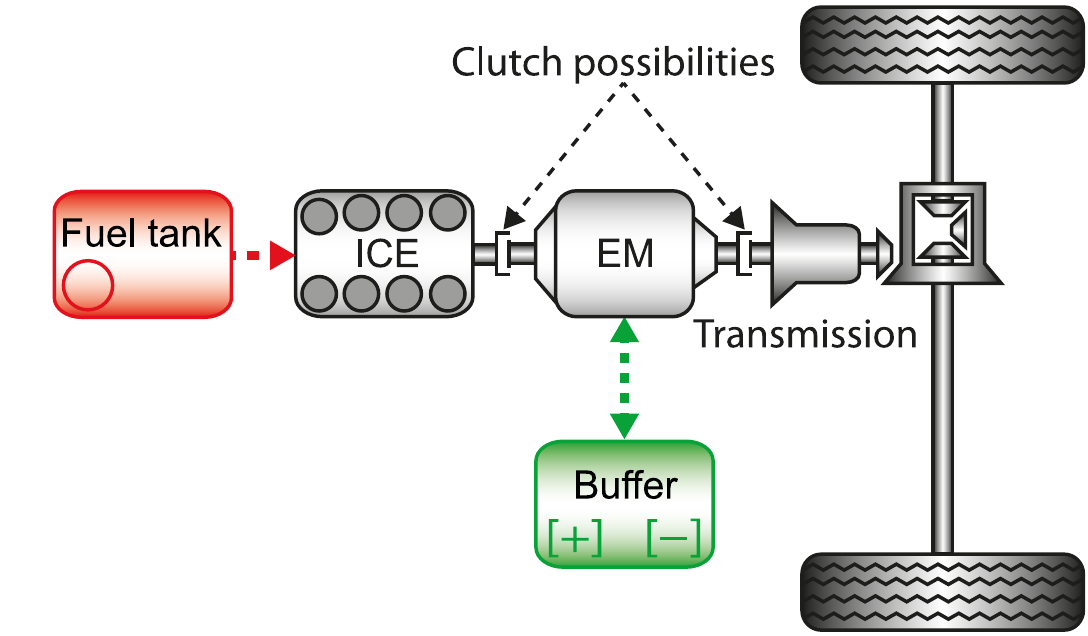}
    \caption{Illustration of a parallel HEV powertrain. The powertrain includes an internal combustion engine (ICE), an electric machine (EM) and electric buffer, i.e. a battery. The ICE and EM are mounted on the same shaft 
    \cite{johannesson15a}.}
    \label{fig:parallelhev}
\end{figure}

Let $\omega$ and $M_\tx{dem}$ denote the speed and total torque delivered at the shaft between the EM and the gearbox. The total torque, which is hereafter referred to as the demanded torque,   
\begin{align}
    M_\tx{dem}(t)=M_\tx{E}(t)+M_\tx{M}(t) + M_\tx{Sbrk}(t) + M_\tx{Abrk}(t)\label{eq:demandedtorque}
\end{align}
    can be satisfied by aggregating the torques from the ICE ($M_\tx{E}$), EM ($M_\tx{M}$), service brakes, 
    \begin{align} \label{eq:service_braking}
        M_\tx{Sbrk}(t) \leq 0
    \end{align}
    and additional braking $M_\tx{Abrk}$,
    \begin{align} \label{eq:additional_braking}
        M_\tx{Amin}(\omega)\leq M_\tx{Abrk}(t)\leq 0
    \end{align}
    achieved by, e.g., a retarder, a compression release engine brake and/or an exhaust pressure governor. A more proper notation of the lower torque limit in \eqref{eq:additional_braking} is $M_\tx{Amin}(\omega(t))$, clearly indicating that the speed $\omega$ is a time dependent signal. Due to compactness, the time dependence will not be explicitly shown for signals that are input arguments to functions in this paper. 
    
    The EM torque, the sign of which can be either positive or negative, draws/recuperates electrical power
\begin {align}
	P_\tx{Mel}(M_\tx{M},\omega) &= d_{0}(\omega) + d_{1} (\omega)M_\tx{M}(t) + d_{2} (\omega) M_\tx{M}^2(t)\label{eq:electricpower}
\end {align}
    from/to the battery. The  coefficients ${d_i(\omega)\geq 0, \forall \omega, i=0,1}$ and ${d_2(\omega)>0, \forall \omega,}$ are obtained by fitting a second order polynomial in torque, for grid values of $\omega$ within the entire speed range of the powertrain. For any other values of $\omega$, $d_i$ are obtained by linear interpolation. This type of model is a common choice for modeling both EMs and ICEs for HEV power-split problems (see e.g. \cite{murgovski12Mec, hu16, hovgard18CEP}).
    The EM torque is bounded by speed dependent limits 
\begin{equation}
    M_\tx{Mmin}(\omega) \leq M_\tx{M}(t) \leq M_\tx{Mmax}(\omega) \label{eq:emtorque}
\end{equation}
    as illustrated in \fig{fig:emmodel}. The figure also depicts the efficiency lines of the EM, computed as the ratio between mechanical power $\omega M_\tx{M}$ and the corresponding electrical power $P_\tx{Mel}$.
\begin{figure}[t]
    \centering
    \includegraphics[width=0.93\linewidth]{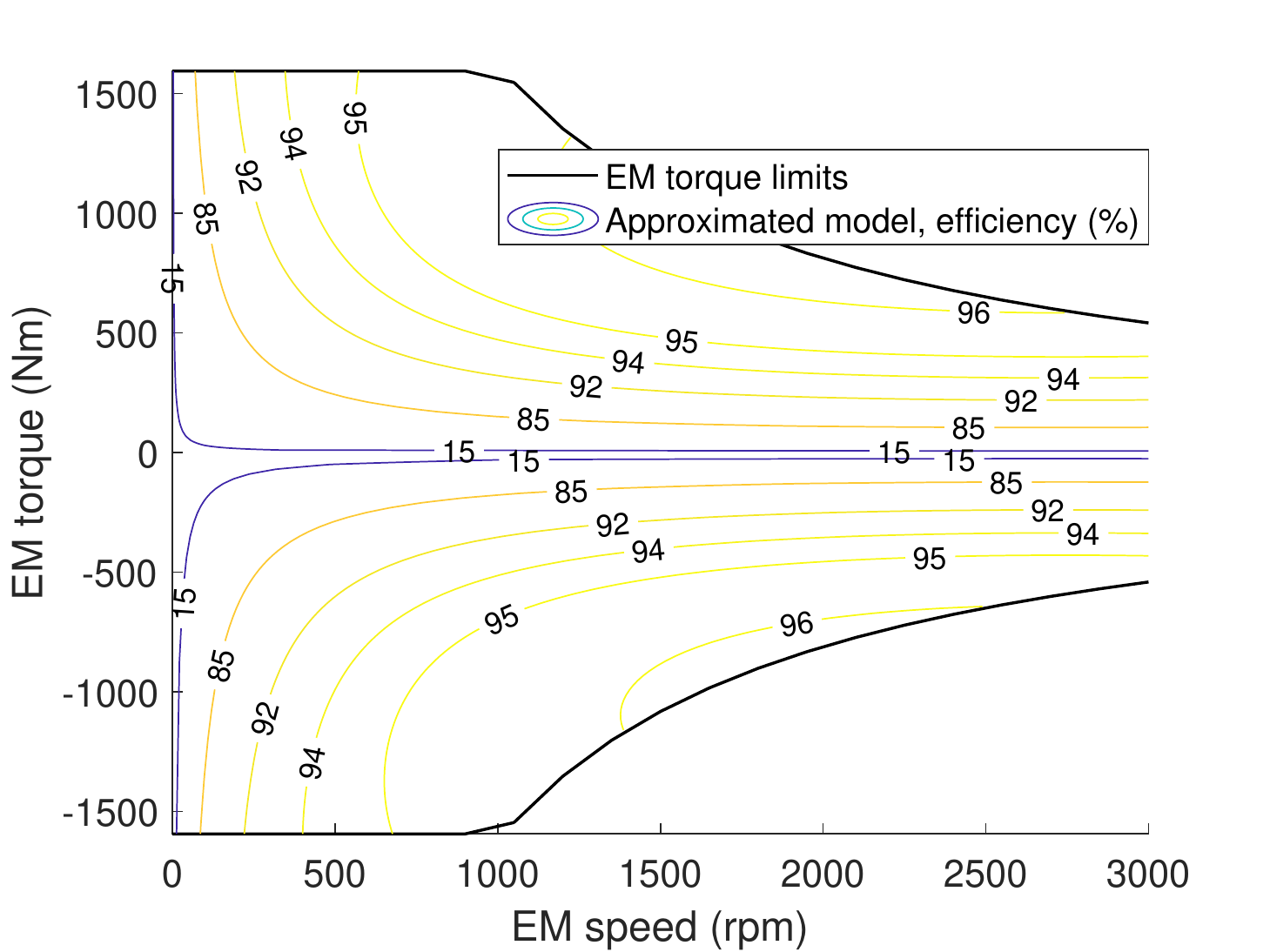}
    \caption{Model of the electric machine, depicting torque limits and efficiency lines over the entire operating range.}
    \label{fig:emmodel}
\end{figure}  
    The power balance equation that relates electrical battery power $P_\tx{Bel}$ and EM electrical power can be expressed as
\begin{equation}
    P_\tx{Bel}(P_\tx{B})=P_\tx{B}(t) - R_\tx{B}\frac{P_\tx{B}^2(t)}{U_\tx{oc}^2} = P_\tx{Mel}(M_\tx{M},\omega) + P_{\tx{aux}} \label{eq:elbalance}
\end{equation}
    where $P_\tx{B}$ is internal (chemical) battery power, $R_\tx{B}$ and $U_\tx{oc}$ are internal battery resistance and open circuit voltage, respectively, and $P_\tx{aux}$ is electrical auxiliary load on the battery \cite{hovgard18CEP}. 
    The battery power is bounded according to
\begin{equation} 
    P_\tx{Bmin}\leq P_\tx{B}(t) \leq P_\tx{Bmax}. \label{eq:pblimits}
\end{equation} 

    The battery dynamics are modeled as 
\begin {equation}
     \dot {\tx{SOC}}(t) = -\frac {P_\tx{B}(t)}{E_\tx{Bmax}} \label{eq:state1}
\end {equation}
    where $E_\tx{Bmax}=Q_\tx{nom}U_\tx{oc}$ is battery energy capacity, $Q_\tx{nom}$ is its nominal capacity and $\tx{SOC}$ is its state of charge, bounded as
\begin{equation}
    0 \leq \tx{SOC}_\tx{min} \leq \tx{SOC}(t) \leq \tx{SOC}_\tx{max} \leq 1.
    \label{eq:soclimit}
\end{equation}

    The fuel consumption by the ICE, $\mu_\tx{fuel}$, can be modeled as a second order function in engine torque
\begin {align}
	\mu_\tx{fuel}(M_\tx{E},\omega)= a_0(\omega) + a_1(\omega) M_{\tx{E}}(t) + a_2(\omega) M_{\tx{E}}^2(t)\label{eq:fuel}
\end{align}
    where the coefficients ${a_i(\omega) \geq 0, \forall \omega, i=0,2}$ and ${a_1(\omega)>0}, \forall \omega$, are found similarly as those of the EM. The ICE torque is bounded 
\begin{equation}
    M_\tx{Emin}(\omega) \leq M_\tx{E}(t) \leq M_\tx{Emax}(\omega) \label{eq:entorque}
\end{equation}
    with speed-dependent torque limits. 
    When the ICE operates with its minimum torque, fuel consumption is zero that ensures the fuel consumption is non-negative.
    The ICE torque limits and efficiency are depicted in \fig{fig:icemodel}. The figure also depicts the additional braking torque and its limits. 
    This torque is defined in the region between $M_\tx{Emin}$ (the solid line) and  $M_\tx{Amin}+M_\tx{Emin}$ (the dashed line) curves. It is the difference between the torques in this region and the ICE torque lower bound.

 \begin{figure}[t]
    \centering
    \includegraphics[width=0.93\linewidth]{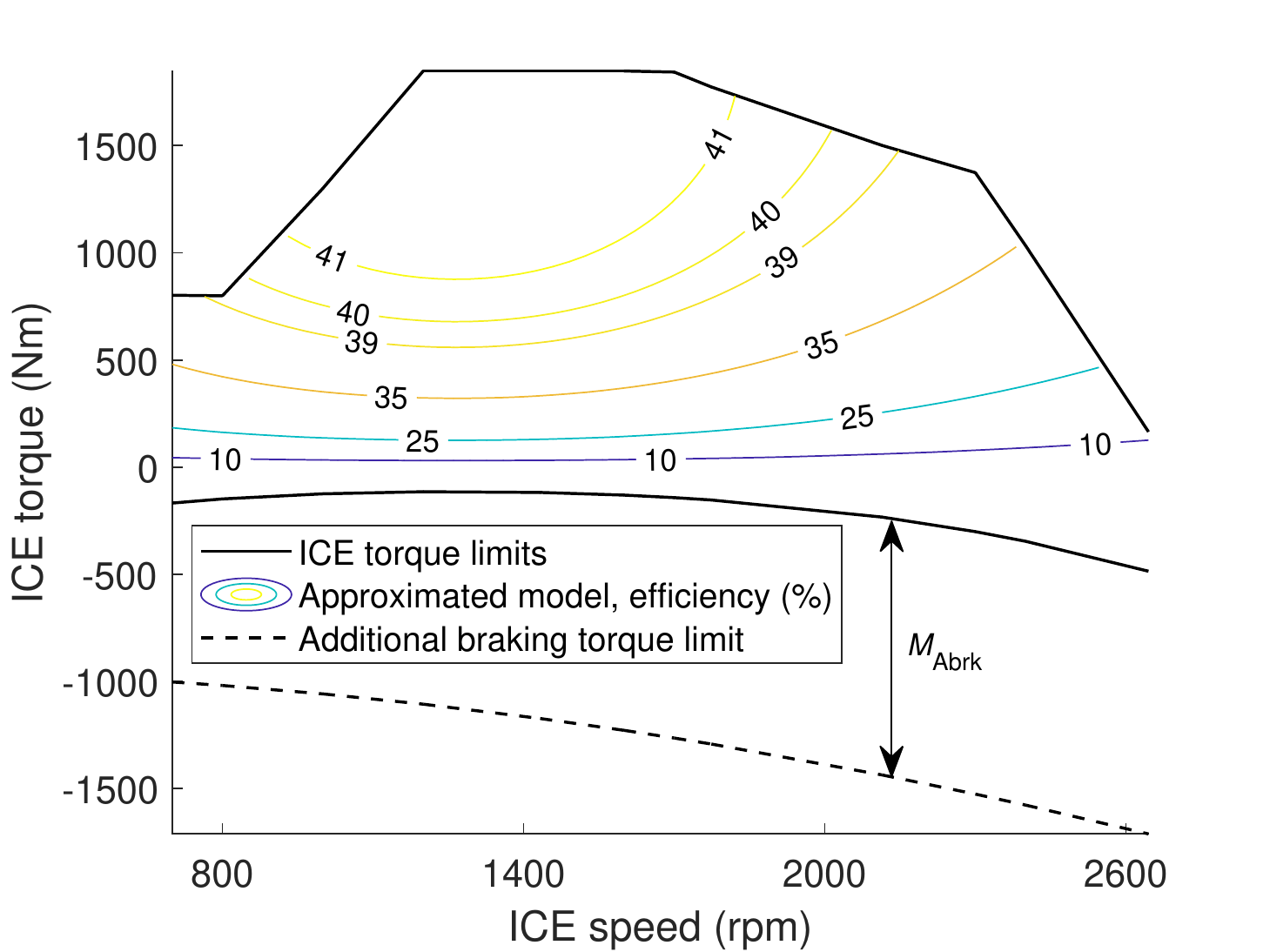}
    \caption{Model of the internal combustion engine depicting torque limits, additional braking torque and efficiency lines over the positive torque region.}
    \label{fig:icemodel}
\end{figure}
\subsection{Functionalities of the local power-split control} \label{sec:objectives}
    
    Situated at the low levels of the multi-layer control architecture, where controllers operate with high frequency, the objective of the local power-split  controller is to deliver computationally efficient decisions that optimally split the requested torque between the ICE, EM and service brakes. 
  
    Let $\check \omega(\tau)$ and $\check M_\tx{dem}(\tau)$ denote the reference trajectories for the ICE/EM speed and torque, respectively, at time instant $\tau$, $\check {\tx{SOC}}(\tau)$ denote the reference battery SOC and $\check \lambda_\tx{B}(\tau)$ its reference costate, which translates battery energy to a fossil-fuel energy \cite{guzzella13}. As mentioned earlier, these trajectories are either generated by a supervisory controller, or are a combination of a driver request and a heuristic predictor. For e.g., these trajectories could be guessed by anticipating the forthcoming energy demands with respect to the road topography.
    
    The aim of the controller is to deliver the demanded torque, or as close to it as possible, minimize fuel consumption and avoid usage of $M_\tx{Sbrk}$ that may unnecessarily wear the service brakes.  
    The final objective is SOC reference tracking, for which two options are considered. The most common and ubiquitous choice in literature is tracking a target SOC at the final time only. To this end, the ECMS strategy is generally employed, where a linear penalty term is adjoined to the objective, i.e. 
\begin{align}
    s_\tx{B}(\tau) P_\tx{B}(t|\tau)
\end{align}
    where ${s_\tx{B}(\tau) = -\lambda_\tx{B}(\tau)/E_\tx{Bmax}}$
    is called the equivalent factor and $\lambda_\tx{B}$ is the battery costate. Here, $\tau$ is the time instant when the controller is evaluated and $t\in[\tau, t_\tx{f}]$ is any time instance along a look-ahead horizon with final time $t_\tx{f}$. During one evaluation of the controller, system parameters and set points are considered time invariant, i.e. the reference signals from the supervisor are assumed to hold their requested value at time $\tau$.  
    
    The ECMS strategy derives directly from the Pontryagin's minimum principle \cite{naidu02}, where the optimal solution of the primal problem can also be obtained by minimizing the Hamiltonian
\begin{align}
\begin{split}
    H(\cdot) = 
    \mu_\tx{fuel}(M_\tx{E},\check\omega) 
    + s_\tx{B}(\tau) P_\tx{B}(t|\tau) \label{eq:ham}  
\end{split}
\end{align}
where a necessary condition for the costate is 
\begin{align} 
    \dot{\lambda^*_{\tx{B}}}(\tau) =-{\left(
    \frac{\partial H(\cdot)}{\partial \tx{SOC}(t|\tau)}\right)}^\tx{*} \label{eq:necessary_condition}
\end{align} 
    with the superscript $*$ denoting the optimal solution. 
    Because battery internal resistance and open circuit voltage are assumed to be constant during one update of the controller, the Hamiltonian is independent of SOC and
\begin{align} 
	\dot{\lambda^*_{\tx{B}}}(\tau) =0 \quad \Rightarrow  \quad \lambda^*_{\tx{B}}(\tau)=\tx{constant}, \label{eq:lambdab}
\end{align}
    i.e. $\lambda^*_{\tx{B}}$ is not a function of time $t$. This strategy is typically used in supervisory controllers, where the reference costate $\check \lambda_\tx{B}$ is obtained by solving a two-point-boundary-value problem that leads the system from an initial state $\tx{SOC}(\tau)$ to a target state \cite{guzzella13,delprat04}. However, this strategy omits the SOC limits \eqref{eq:soclimit}, since \eqref{eq:necessary_condition} is not defined when SOC is on the bound. In fact, when SOC is on the bound, the costate is free to jump to other values and, hence, it may not be a constant. One way to resolve this issue is to use a SOC feedback in the local controller to adapt $\lambda_\tx{B}$ as a function of the reference costate $\check\lambda_\tx{B}$ and $\tx{SOC}(\tau)$ \cite{pisu07}. The adaptive strategy is discussed further in Section \ref{sec:adaptive_ECMS}.

    The other option for SOC reference tracking is a straightforward cost term  
\begin{align} 
   q_\tx{SOC}\, (\tx{SOC}(t|\tau) - \bar{\tx{SOC}}(\tau))^2 \label{eq:soc_tracking}
\end{align}
    which directly penalizes deviations from the saturated reference SOC \cite{tianheng15},
    \begin{align*}
    \bar{\tx{SOC}} = \max(\min(\check{\tx{SOC}}, \tx{SOC}_\tx{max}-\epsilon), \tx{SOC}_\tx{min}+\epsilon)
\end{align*}
    using the penalty coefficient $q_\tx{SOC}$, 
    where $\epsilon > 0$ is a small positive value that serves two purposes. First, to provide further incentive for keeping the battery SOC away from the bounds, and second, to provide robustness towards uncertainty in measured SOC value. Further discussion on the robustness is provided in Section~\ref{sec:robustness}.
\subsection{Problem statement}   
    When designing a computationally efficient controller, it is suitable to formulate the problem using deliverable demanded torque, reduce the number of control inputs, relax the SOC constraints, seek convexity, and possibly obtain an analytic solution that can be evaluated efficiently. 
    To  state  the  problem, time 
    is discretized with zero order hold. For a chosen sampling interval $T_\tx{s}$, discretization will give 
    time instances ${\tau=jT_\tx{s}, j=0,1,\dots}$.
    The time when the controller is updated is denoted by $\tau$, 
    while the time instant along the prediction horizon of each update is denoted by 
    ${t=\tau + kT_\tx{s}}$, ${k=0,1,\dots}$. In this paper we study the case where the prediction horizon could be infinite.

    Let $x=\tx{SOC}$ denote the single state in the problem, ${\check x=\bar{\tx{SOC}}}$ its saturated reference, ${x_\tx{min}=\tx{SOC}_\tx{min}}$ and ${x_\tx{max}=\tx{SOC}_\tx{max}}$ its bounds 
    and $u=P_\tx{B}$ the battery power. By applying zero-order hold, the system dynamics \eqref{eq:state1} can be written in a discrete form  
    \begin{align}
        x(k+1|j) = x(k|j)-\frac {T_\tx{s}}{E_\tx{Bmax}}u(k|j).
        \label{eq:discrete_dynamics}
    \end{align}
    These are the necessary ingredients to state the goals of this paper, as follows:
\begin{enumerate}
    \item Derive a saturation function
\begin {align} \label{eq:saturation_function}
    \bar M_\tx{dem}(j)=\sat(\check \omega,\check M_\tx{dem})
\end{align}
    that computes the deliverable demanded torque $\bar M_\tx{dem}$, for a torque $\check M_\tx{dem}$ requested by the driver or supervisory controllers. 
    \item Show that a single control input
    \begin{align} \label{eq:control_bounds}
    u(k|j)=P_\tx{B}(k|j) \in [u_\tx{min}(\check\omega), u_\tx{max}(\check\omega)]
    \end{align}
    is sufficient to derive the torques of all the actuators,
    \begin{align}
    \begin{split} \label{eq:f_trq_function}
        [M_\tx{E}(k|j),M_\tx{M}(k|j),M_\tx{Abrk}(k|j),&M_\tx{Sbrk}(k|j)]^T\\
        &= f_\tx{M}(u)
    \end{split}
    \end{align}
    such that usage of $M_\tx{Sbrk}$ is avoided, and the feasible set of the control input is defined as
    {
    \begin{align}
    \begin{split}
        [u_\tx{min}(\check\omega),& u_\tx{max}(\check\omega)] = \Big\{ [\cdot]=f_\tx{M}(u)\,| \\
    &\bar M_\tx{dem}(k|j)=M_\tx{E}(k|j)+M_\tx{M}(k|j)\\
    &\hspace{10mm}+ M_\tx{Sbrk}(k|j) + M_\tx{Abrk}(k|j),\\
    &M_\tx{Emin}(\check\omega) \leq M_\tx{E}(k|j) \leq M_\tx{Emax}(\check\omega),\\
    &M_\tx{Mmin}(\check\omega) \leq M_\tx{M}(k|j) \leq M_\tx{Mmax}(\check\omega),\\
    &M_\tx{Amin}(\check\omega) \leq M_\tx{Abrk}(k|j) \leq 0,\\
    &M_\tx{Sbrk}(k|j) \leq 0,\\
    &P_\tx{Bmin} \leq u(k|j) \leq P_\tx{Bmax},\\
    &x_\tx{min} \leq x(1|j) \leq x_\tx{max}, \quad
    \forall k, \quad \forall j \Big\}
    \end{split}
    \label{eq:constraints}
    \end{align}}%
    where the constraints on SOC  
    are enforced only at the instant $k=1$. At remaining instances, SOC constraints are enforced by an interior point penalty function. The reason for doing this will be described later, in Section~\ref{sec:deliverable_demand}.
    \item Derive the fuel consumption function $\mu_\tx{fuel}$ in terms of $u$ and prove convexity of two optimal controllers, one based on ECMS and one on the SOC tracking \eqref{eq:soc_tracking}. The problem relying on ECMS is stated as
\begin{subequations} \label{eq:op_ECMS} 
\begin {align}  
    \min_u\; &\mu_\tx{fuel}(\check\omega,u) + s_\tx{B}(j)u(j)\\
    \text{s.t.: } & u(j)\in [u_\tx{min}(\check\omega),  u_\tx{max}(\check\omega)], \forall j
\end{align}
\end{subequations}
    with 
\begin{align}
    \begin{split}
    &s_\tx{B}(j) = f^\tx{ECMS}_\tx{p}(\check s_\tx{B},\check x,x_\tx{m}),\\
    &\check s_\tx{B}(j)=-\check\lambda_\tx{B}(j)/E_\tx{Bmax}
    \end{split} \label{eq:adaptionrule}
\end{align}
    where $f^\tx{ECMS}_\tx{p}$ is a nonlinear interior point penalty function that aims to keep the SOC within bounds.
    
    The problem with SOC reference tracking is stated as
{\allowdisplaybreaks
\begin{subequations} \label{eq:op_LQT}    
\begin {align}
    \begin{split}
     \min_u & \sum_{k=0}^{\infty} \mu_\tx{fuel}(\check\omega,u) + q_\tx{SOC}\,(x(k|j)-\check x(j))^2\\
     &\hspace{40mm}+ f^\tx{LQT}_\tx{p}(\check x, x)
    \end{split}\\
    \tx{s.t.: } &x(k+1|j) = x(k|j) -\frac{T_\tx{s}u(k|j)}{E_\tx{Bmax}} , \forall k, \forall j \label{eq:statespace}\\
    & u(k|j)\in [u_\tx{min}(\check\omega), u_\tx{max}(\check\omega)], \forall k, \forall j\\
    &x(0|j)=x_\tx{m}(j), \forall j
    \label{eq:initialstate}
\end{align}
\end{subequations}}%
    where $x_\tx{m}$ is the measured SOC value at instant ${\tau=jT_\tx{s}}$ and $f^\tx{LQT}_\tx{p}$ is a strictly convex, nonlinear interior point penalty function that keeps SOC within bounds. 
    In this problem, an infinite horizon controller is designed and updated in every time instance $j$. This is because the demanded torque, the reference value of ICE/EM speed and reference SOC can change at every instant.
    \item Propose a computationally efficient sub-optimal solution of problems \eqref{eq:op_ECMS} and \eqref{eq:op_LQT}, by deriving an analytic solution of their second order approximation.  
\end{enumerate}

To summarize, the goals of this paper are to derive the function $\sat$ in \eqref{eq:saturation_function}, the control bounds $u_\tx{min}$ and $u_\tx{max}$ in \eqref{eq:control_bounds} and the function $f_\tx{M}$ in \eqref{eq:f_trq_function}, propose the interior point functions $f^\tx{ECMS}_\tx{p}$ and $f^\tx{LQT}_\tx{p}$, prove convexity of problems \eqref{eq:op_ECMS} and \eqref{eq:op_LQT}, and provide a sub-optimal analytic solution.

For didactic reasons and for compactness, the dependence on $j$ will not be shown in the rest of the paper.
\section{Compact convex formulation}         \label{sec:convex_problem}     
    In this section, we derive the saturation function in \eqref{eq:saturation_function} for calculating the deliverable demanded torque, show how the reduction of control inputs can be achieved  
    and derive the feasible set of the control input in \eqref{eq:constraints}, and then the functions for obtaining the torques of all the actuators such that usage of service brakes is avoided. 
    We also show that problems \eqref{eq:op_ECMS} and \eqref{eq:op_LQT} are nonlinear convex programs.
\subsection{Deliverable demand and optimal actuator torques} \label{sec:deliverable_demand}
    Recall that at a given update of the controller at time instant $j$, for all instances $k$ along the horizon, the demanded torque $\check M_\tx{dem}$ is assumed to be holding its value requested at instant $j$.
\begin{lemma}
    A necessary condition ensuring that demanded torque is delivered for an infinite horizon is 
\begin{align}
     \check M_\tx{dem} \leq M_\tx{Emax}(\check\omega) + M_\tx{Meq}(\check\omega)
     \label{eq:deliverable_torque_infinite_horizon}
    \end{align} 
    where 
\begin{align*}
        M_\tx{Meq}(\check\omega) = \frac{-d_{1}(\check\omega)}{2d_{2}(\check\omega) }
    + \frac{\sqrt{ d_{1}^2(\check\omega)-4d_{2}(\check\omega) (d_{0}(\check\omega)+P_\tx{aux}) } }{2d_{2}(\check\omega) }
\end{align*}
    is the EM torque required to run the electrical auxiliaries, i.e. system \eqref{eq:discrete_dynamics} is at equilibrium.
\end{lemma}
\begin{proof}
    The proof follows directly by investigating the steady-state condition in \eqref{eq:discrete_dynamics} in the case of infinite horizon control, which requires battery power to be zero, requiring the EM to operate with torque $M_\tx{Meq}$ sufficient to power the auxiliaries. 
    This torque is calculated by using  \eqref{eq:electricpower} and \eqref{eq:elbalance}, see Lemma~\ref{lemma:em_torque} in Appendix \ref{sec:lemmas}.
    Then, the maximum torque can be delivered if the ICE operates with its maximum torque. 
\end{proof}

    It is clear that condition \eqref{eq:deliverable_torque_infinite_horizon} is very conservative and impractical, as it merely states that the battery should not be used when demanded torque is greater than ${M_\tx{Emax} + M_\tx{Meq}}$. This would disable one of the advantages of HEVs, to boost the acceleration performance of the vehicle by delivering power from both the ICE and EM when demanded torque is high, so in the rest of the paper we will not commit to ensuring the deliverance of demanded torque for an infinite horizon. Instead, we focus on delivering demanded torque at the current instant $k=0$, such that battery state stays within bounds at the next instant $k=1$. Since the controller is re-evaluated at every instant, persistent feasibility is guaranteed.
\begin{proposition} 
    Control bounds on the battery power, $u_\tx{min}$, $u_\tx{max}$, can be derived such that, at a given time instant $k$, the delivered torque $\bar M_\tx{dem}(k)$ is as close as possible to the demanded torque $\check M_\tx{dem}(k)$, without violating the constraints in \eqref{eq:constraints}. Furthermore, the deliverable demanded torque can be computed as
\begin{align}\label{eq:delivered_torque}
\begin{split}
    \bar M_\tx{dem}(\check\omega) &=\sat(\check \omega,\check M_\tx{dem})\\
    &=\min (\check M_\tx{dem},  M_\tx{Emax}(\check\omega) + \bar M_\tx{Mmax}(\check\omega))
\end{split}
\end{align}
    where maximum deliverable EM torque is
\begin{align}
\begin{split}
    &\bar M_\tx{Mmax}(\check\omega)= \frac{-d_{1}(\check\omega)}{2d_{2}(\check\omega) }\\
    &\hspace{5mm} + \frac{\sqrt{ d_{1}^2(\check\omega)-4d_{2}(\check\omega) (d_{0}(\check\omega)-P_\tx{Bel}(u_\tx{max})+P_\tx{aux}) } }{2d_{2}(\check\omega) }
\end{split} \label{eq:deliverable_EM_torque}
\end{align}
    and $P_\tx{Bel}$, expressed directly from \eqref{eq:elbalance}, is
    \[ P_\tx{Bel}(u) = P_\tx{Mel}(\check\omega, M_\tx{M}) + P_\tx{aux}. \]
\end{proposition}
\begin{proof}
    We will prove the proposition by deriving the control bounds $u_\tx{min}$ and $u_\tx{max}$ and by analyzing the torque balance when demanded torque is positive or negative. 

    If demanded torque is negative, all actuators can be used to deliver part of it. In this case, $M_\tx{Sbrk}$ is used as the last resort and only when the other actuators have reached their minimum torque limit. Similarly, there is no cost incentive for using the additional braking by retarder or compression release engine brake, so $M_\tx{Abrk}$ will also be avoided if not necessary. Hence, the ICE and/or EM will be used to brake the vehicle until they reach their minimum torque limit. 

    From all the actuators, only the ICE and EM can be used to deliver positive demanded torque. When demanded torque is higher than what the ICE and EM can deliver, i.e. ${\check M_\tx{dem} \geq M_\tx{Emax} + \bar M_\tx{Mmax}}$, the ICE will operate with maximum torque, to deliver torque that is as close as possible to the demanded torque. In that case, the EM will either operate with maximum torque, or a battery power or SOC limit will get activated. Then, the ICE torque limits can be reflected on the EM, as
\begin{align}
    \begin{split}
        \hat M_\tx{Mmin}(\check \omega)&= \max\Big(-\frac{d_1(\check\omega)}{2 d_2(\check\omega)}, M_\tx{Mmin}(\check \omega),\\
        & \min\left(M_\tx{Mmax}(\check \omega), \check M_\tx{dem} - M_\tx{Emax}(\check \omega)\right) \Big) 
        \label{eq:EM_torque_min}
    \end{split}\\
    \begin{split}
        \hat M_\tx{Mmax}(\check \omega)&= \min\Big(M_\tx{Mmax}(\check \omega), \\
        &\max\left(\hat M_\tx{Mmin}(\check \omega), \check M_\tx{dem} - M_\tx{Emin}(\check \omega) \right)\Big) 
        \label{eq:EM_torque_max}
    \end{split}
\end{align}
    where the additional implicit torque limit $-d_1/(2d_2)$ arises from the minimum achievable electrical power, see Lemma~\ref{lemma:em_torque} in Appendix \ref{sec:lemmas}. 
    Since the EM is an incrementally passive component, its electrical power increases with its torque. Similar relation holds between the chemical and electrical power of the battery, see Lemma \ref{lemma:em_torque} and \ref{lemma:batterypower} in Appendix \ref{sec:lemmas}. Hence, the EM torque limits can be translated as limits on the EM electrical power, which in turn can be reflected as limits on the battery chemical power, 
{
\begin{align}
\begin{split}
    &u_\tx{min}(\check\omega)=\max \Bigg( P_\tx{Bmin}, \frac{(x_\tx{m}-x_\tx{max}+\epsilon)E_\tx{Bmax}}{T_\tx{s}},\\
    &\hspace{3mm}\min\Bigg( \frac{U_\tx{oc}^2}{2 R_\tx{B}}, P_\tx{Bmax}, \frac{(x_\tx{m}-x_\tx{min}-\epsilon)E_\tx{Bmax}}{T_\tx{s}}, \\
    &\hspace{3mm}\frac{U_\tx{oc}^2 - U_\tx{oc}\sqrt{U_\tx{oc}^2 - 4R_\tx{B} (P_\tx{Mel}(\check\omega, \hat M_\tx{Mmin}) +P_\tx{aux})} }{2R_\tx{B}} \Bigg) \Bigg)
\end{split}\label{eq:u_min}\\
\begin{split}
    &u_\tx{max}(\check\omega)=\min \Bigg( \frac{U_\tx{oc}^2}{2 R_\tx{B}}, P_\tx{Bmax},\\
    &\hspace{3mm}\frac{(x_\tx{m}-x_\tx{min}-\epsilon)E_\tx{Bmax}}{T_\tx{s}}, \max\Bigg(u_\tx{min}(\check\omega),  \\
    &\hspace{3mm}\frac{U_\tx{oc}^2 - U_\tx{oc}\sqrt{U_\tx{oc}^2 - 4R_\tx{B} (P_\tx{Mel}(\check\omega, \hat M_\tx{Mmax}) +P_\tx{aux})} }{2R_\tx{B}}\Bigg) \Bigg)
\end{split} \label{eq:u_max}
\end{align}}%
    where an implicit limit $U_\tx{oc}^2/(2R_\tx{B})$ has also been included, see Lemma \ref{lemma:batterypower} in Appendix \ref{sec:lemmas}. The SOC limits \eqref{eq:soclimit} have been reflected as limits on the control input by deriving the battery power from \eqref{eq:discrete_dynamics} corresponding to when $x$ at the following instant is at one of its limits. 

    The deliverable demanded torque can now be obtained by reversing the steps used for obtaining the upper control bound. By substituting the power balance \eqref{eq:elbalance} in \eqref{eq:u_max} and using Lemma \ref{lemma:em_torque} in Appendix \ref{sec:lemmas}, the maximum deliverable EM torque can be obtained exactly as in \eqref{eq:deliverable_EM_torque}. Hence, the deliverable demanded torque cannot exceed the sum of the maximum ICE and deliverable EM torque and then \eqref{eq:delivered_torque} addresses the saturation function in \eqref{eq:saturation_function}. 
\end{proof}

    The optimal torques of the actuators and therefore $f_\tx{M}$ in \eqref{eq:f_trq_function} are obtained in a similar manner using the following proposition.
\begin{proposition} \label{proposition:optimaltorques}
    Let $u^*\in [u_\tx{min},u_\tx{max} ]$ denote the optimal control input. Then, the optimal actuator torques can be obtained as $[M_\tx{M}^*,M_\tx{E}^*,M_\tx{Abrk}^*,M_\tx{Sbrk}^*]^T = f_\tx{M}(u^*)$, where the function $f_\tx{M}$ is implemented as
    {\allowdisplaybreaks
\begin{subequations} 
\begin {align}
    &M_\tx{M}^*= \frac{-d_{1} +\sqrt{ d_{1}^2-4d_{2} (d_{0}-P_\tx{Bel}(u^*)+P_\tx{aux}) } }{2d_{2} } \label{eq:optimal_EM_torque}\\
    &M_\tx{E}^* = \max(M_\tx{Emin}(\check\omega), \bar M_\tx{dem}(\check\omega) - M_\tx{M}^*) \label{eq:optimal_ICE_torque}\\
    &M_\tx{Abrk}^*=\max(M_\tx{Amin}(\check\omega), \bar M_\tx{dem}(\check\omega) - M_\tx{M}^* - M_\tx{E}^*) \label{eq:optimal_add_brake}\\
    &M_\tx{Sbrk}^*= \bar M_\tx{dem}(\check\omega) - M_\tx{M}^* - M_\tx{E}^* - M_\tx{Abrk}^* \label{eq:optimal_service_brake}.
\end {align}
\end{subequations}}%
\label{propos:torques}
\end{proposition}
\begin{proof}
    The optimal EM torque is obtained exactly as in \eqref{eq:deliverable_EM_torque}, by simply replacing $u_\tx{max}$ with $u^*$. If demanded torque is positive, then the optimal ICE torque is the difference between the deliverable demanded torque and the optimal EM torque. If demanded torque is negative, the ICE will provide additional braking down to its lower torque limit. After that $M_\tx{Abrk}$ is used to its saturation limit, and finally $M_\tx{Sbrk}$ is used to deliver the remaining demand.
\end{proof}


\subsection{Problem convexity}
    Let 
\begin{align}
     M_\tx{M}(\check\omega,u) = e_0(\check\omega) + \sqrt{e_1(\check\omega) + e_2(\check\omega) P_\tx{Bel}(u)} \label{eq:EM_torque_concave}
\end{align}
    with $e_2=1/d_2 > 0$, denote the EM torque expressed as a function of battery power, as in \eqref{eq:optimal_EM_torque} (see also Lemma~\ref{lemma:em_torque} in Appendix \ref{sec:lemmas}). Then, by following \eqref{eq:optimal_ICE_torque}, ICE torque may be expressed as
\begin{align}
    M_\tx{E}(\check\omega,u) = \max(M_\tx{Emin}(\check\omega), \bar M_\tx{dem}(\check\omega) - M_\tx{M}(\check\omega,u)) \label{eq:ICE_torque_convex}
\end{align}
which allows fuel consumption to be written as a function of battery power, i.e. ${\mu_\tx{fuel}(\check\omega,u) = \mu_\tx{fuel}(\check\omega, M_\tx{E}(\check\omega,u))}$. 

\begin{lemma}
    The fuel consumption $\mu_\tx{fuel}(\check\omega,u)$ is a convex and monotonically decreasing function in $u$.
    \label{lemma:convexity}
\end{lemma}
\begin{proof}
    From the electric power balance relation \eqref{eq:elbalance} it is clear that $P_\tx{Bel}(u)$ is a concave function, and from the proof of Lemma~\ref{lemma:batterypower} (see Appendix \ref{sec:lemmas}) it is strictly monotonically increasing in $u$. Furthermore, since the square-root function in \eqref{eq:EM_torque_concave} is concave and non-decreasing, $M_\tx{M}(\check\omega,u)$ is also concave in $u$, and also strictly monotonically increasing (see the proof of  Lemma~\ref{lemma:em_torque} in Appendix \ref{sec:lemmas} and the composition rules for convexity \cite[p.~83]{boyd04}). The negative of $M_\tx{M}(\check\omega,u)$ is convex and strictly monotonically decreasing, and the maximum of convex functions in \eqref{eq:ICE_torque_convex} is a convex function, so $M_\tx{E}(\check\omega,u)$ is a convex function in $u$, and also monotonically decreasing. Since the fuel consumption function \eqref{eq:fuel} is convex and strictly monotonically increasing in $M_\tx{E}$, it follows that $\mu_\tx{fuel}(\check\omega,u)$ is a convex and monotonically decreasing function in $u$.
\end{proof}

\begin{theorem} 
    Problems \eqref{eq:op_ECMS} and \eqref{eq:op_LQT} are convex programs.
    \label{theorem:convexity}
\end{theorem}    
\begin{proof} 
    The objective function of problem \eqref{eq:op_ECMS} includes only the fuel consumption $\mu_\tx{fuel}(\check\omega,u)$ and a linear term in $u$, and it is therefore a convex program. The objective function in problem \eqref{eq:op_LQT} includes the fuel consumption $\mu_\tx{fuel}(\check\omega,u)$, a quadratic function that is convex in the state and a penalty term $f_\tx{p}^\tx{LQT}$ that is also convex in the state, by definition. Hence the objective function of problem \eqref{eq:op_LQT} is convex in both $u$ and $x$. Problem \eqref{eq:op_LQT} is subject to affine constraints in $x$ and $u$ and it is, therefore, a convex program. 
\end{proof}
\section{Adaptive power-split control} \label{sec:controller}    
    In this section, we propose quadratisation of the objective function and derive an adaptive proportional ECMS and an LQT for calculating a sub-optimal solution of problems \eqref{eq:op_ECMS} and \eqref{eq:op_LQT}.
\subsection{Design of adaptive proportional ECMS} \label{sec:adaptive_ECMS}
    In the formulation of problem \eqref{eq:op_ECMS} the constraint on the battery SOC is modeled as a soft constraint using an interior point penalty function. 
    When designing an adaptive proportional ECMS, the penalty function $f^\tx{ECMS}_\tx{p}$ should adapt the battery equivalent factor as in \eqref{eq:adaptionrule}, such that battery SOC is kept within bounds. In order to design the penalty function, we will use the following fact.

\begin{lemma}
    The battery equivalent factor $s_\tx{B}$ is nonnegative.
    \label{lemma:s_b}
\end{lemma}
\begin{proof}
    There are two possibilities for the optimal solution of problem  \eqref{eq:op_ECMS}. If the optimal control is not on the bounds, it holds
\begin{align} 
	\frac{\partial \mu_\tx{fuel}}{\partial u}+s_\tx{B}=0. \label{eq:hamiltoniandiff}
\end{align}
    According to Lemma~\ref{lemma:convexity}, $\mu_\tx{fuel}$ is monotonically decreasing in $u$, so $\frac{\partial \mu_\tx{fuel}}{\partial u}\leq 0$ and from \eqref{eq:hamiltoniandiff} it follows $s_\tx{B}\geq 0$. Moreover, when $s_\tx{B}=0$, it can be obtained 
\begin{align*}
     \left.\frac{\partial \mu_\tx{fuel}}{\partial u}\right\rvert_{u^o}=&\frac{(2R_\tx{B}u^o-U_\tx{oc}^2)(a_1+2a_2M_\tx{E}(\check\omega,u^o))}{U_\tx{oc}^2 \sqrt{d_{1}^2-4d_{2} (d_{0} +P_\tx{aux}-u^o+\frac{R_\tx{B}}{U_\tx{oc}^2}{u^o}^2)}}=0,
\end{align*}   
    which follows from \eqref{eq:fuel}, \eqref{eq:ICE_torque_convex} and Lemma \ref{lemma:em_torque}.
    This gives the optimum
\begin {align}
    u^o= \frac{U_\tx{oc}^2}{2R_\tx{B}} \geq u_\tx{max} \label{eq:u_o}
\end {align} 
    that will, in fact, be saturated to its maximum value $u_\tx{max}$, according to Lemma \ref{lemma:batterypower}, and will require discharging the battery with maximum power. According to Lemma~\ref{lemma:convexity}, $\mu_\tx{fuel}$ is convex in $u$, so $\frac{\partial \mu_\tx{fuel}}{\partial u}$ is non-decreasing and by considering \eqref{eq:hamiltoniandiff}, $s_\tx{B}$ is non-increasing in $u$. Hence, there is no reason for using negative values for $s_\tx{B}$, since the constrained optimal solution is ${u^*=u_\tx{max}}$ when $s_\tx{B}=0$.
\end{proof}

\begin{lemma} \label{lemma:ECMS}
    Let ${s_\tx{B}=f^\tx{ECMS}_\tx{p}(\check s_\tx{B},\check x, x)}$ be an interior point penalty function that satisfies 
\begin{subequations}  \label{eq:condition} 
\begin{align}
    &f^\tx{ECMS}_\tx{p}(\check s_\tx{B},\check x, \check x) =\check s_\tx{B},\\ 
    &f^\tx{ECMS}_\tx{p}(\check s_\tx{B},\check x,x_\tx{max}) =0,\\
    &\lim_{x \rightarrow x_\tx{min}} f^\tx{ECMS}_\tx{p}(\check s_\tx{B},\check x,x)= \infty.
\end{align} 
\end{subequations}%
    Then, functions
{\allowdisplaybreaks
\begin{subequations}  \label{eq:ECMS_penalty_functions} 
\begin{align}
    \begin{split}f^\tx{ECMS}_\tx{p1}(\cdot)&=\max\Bigg(0, \\
    &\check s_\tx{B}+  K_\tx{p1}\bigg(\tan\left(\frac{\pi}{2}\frac{x_\tx{max} + x_\tx{min}-2x_\tx{m}}{x_\tx{max} - x_\tx{min}}\right)\\
    &-\tan\left(\frac{\pi}{2}\frac{x_\tx{max} + x_\tx{min}-2\check{x}}{x_\tx{max} - x_\tx{min}}\right) \bigg)\Bigg),
    \end{split}\\
    \begin{split}f^\tx{ECMS}_\tx{p2}(\cdot)&=\max\Bigg(0, \\
    &\check s_\tx{B} + K_\tx{p2}\bigg((x_\tx{max}-\check{x})\log\left(\frac{x_\tx{max}-x_\tx{m}}{x_\tx{max}-\check{x}}\right)\\
   &-(\check{x}-x_\tx{min})\log\left(\frac{x_\tx{m}-x_\tx{min}}{\check{x}-x_\tx{min}}\right) \bigg)\Bigg)
   \end{split}
\end{align} 
\end{subequations}}%
with $K_\tx{pi}>0, i=1, 2$, are such interior point functions.
\end{lemma}

\begin{proof}
    The proof is straightforward to verify by replacing $x$ with $\check x$, $x_\tx{min}$ or $x_\tx{max}$ into the penalty functions, and computing the equivalent factor according to \eqref{eq:adaptionrule}.
\end{proof}

An illustration of the ECMS penalty functions is provided in Fig.~\ref{fig:ECMS_penalty}.
\begin{figure}[tbp]
    \centering
    \includegraphics[width=0.93\linewidth]{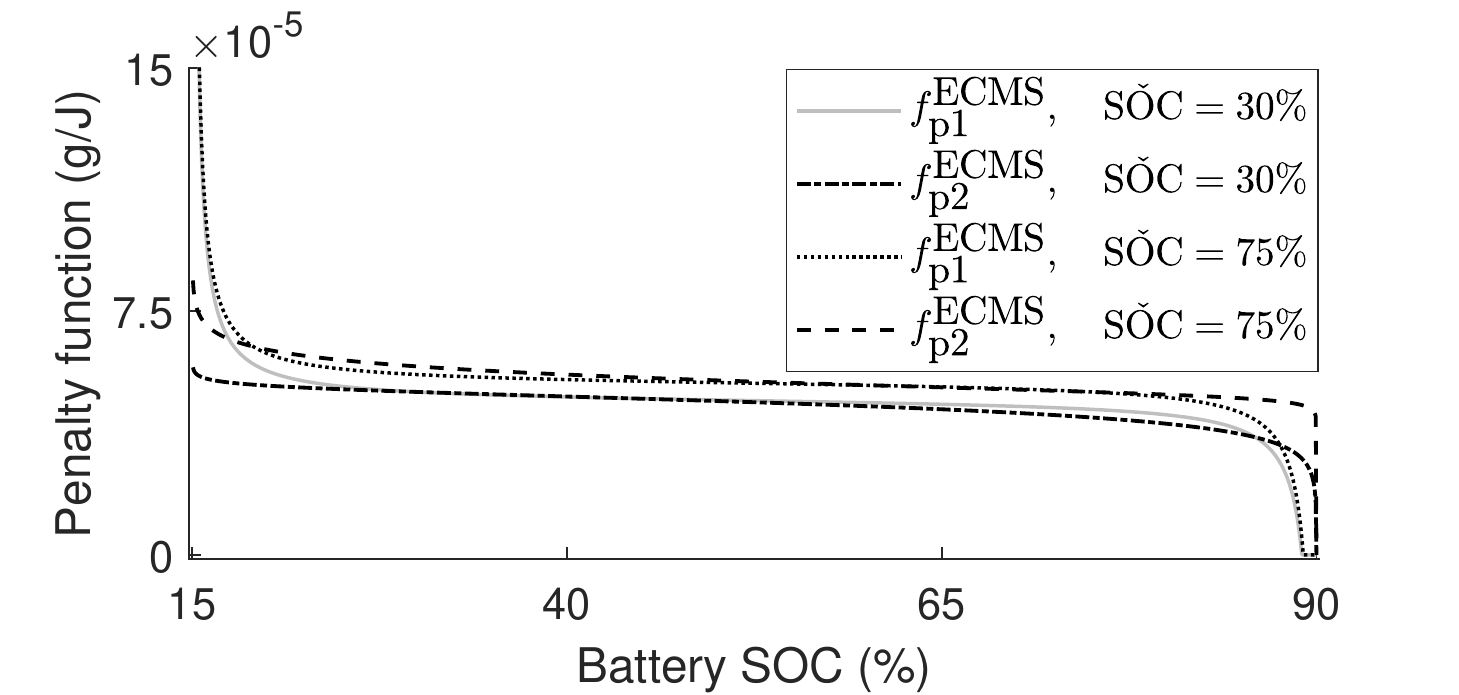}
    \caption{Penalty functions used in problem \eqref{eq:op_ECMS} for two different values of reference SOC. Battery SOC is allowed to vary within 15-\SI{90}{\%}.}
    \label{fig:ECMS_penalty}
\end{figure}
\begin{remark}
    A common choice when designing adaptive ECMS controllers is using a proportional nonlinear feedback, see e.g. \cite{uebel19,rutquist08}. Alternative approaches exist where an integral feedback is also considered, see e.g. \cite{pisu07, johannesson15a}.
\end{remark}
    
    In this paper, the equivalent factor is adapted at every instant using $f^\tx{ECMS}_\tx{p1}$ and $f^\tx{ECMS}_\tx{p2}$ according to  measured SOC and reference values of battery SOC and equivalent factor.
    After the equivalent factor is calculated, the unconstrained problem \eqref{eq:op_ECMS} can be solved analytically, which involves a solution to a quartic equation \cite{hovgard18CEP}. However, obtaining such solutions is complicated, so we suggest approximating the fuel consumption by a quadratic function of $u$, 
\begin {equation}
\begin{split}
    \tilde \mu_\tx{fuel} &= \left. \frac{\partial  \mu_\tx{fuel}}{\partial u} \right\rvert_{u_0}(u-u_0) + \left. \frac{\partial^2  \mu_\tx{fuel}}{2\partial u^2} \right\rvert_{u_0}(u-u_0)^2 \\
    &+ \mu_\tx{fuel}(u_0) =\tilde{a}_{0} + \tilde{a}_{1} (u-u_0) + \tilde{a}_{2} (u-u_0)^2 \label{eq:taylor2}
\end{split}
\end {equation}
    by performing a second order Taylor approximation about 
    ${u_0=u_\tx{max}(\check \omega)}$. The coefficients $\tilde{a}_i, i=0,1,2$, can be found by using \eqref{eq:fuel},  \eqref{eq:ICE_torque_convex} and Lemma \ref{lemma:em_torque}. 

    The reason for choosing $u_0=u_\tx{max}$ is because according to Lemma \ref{lemma:convexity} and definition \eqref{eq:EM_torque_max} and \eqref{eq:u_max}, fuel consumption is monotonically decreasing in $u$, and its minimum value is obtained precisely when $u_0=u_\tx{max}$. Then, according to \eqref{eq:hamiltoniandiff}, the unconstrained sub-optimal solution is calculated as
\begin {align}
    u^o = \frac{2\tilde{a}_2u_0-\tilde{a}_{1}-s_\tx{B}}{2\tilde{a}_2}
\end {align}
    while the sub-optimal solution of the problem \eqref{eq:op_ECMS} is  
 \begin{align} 
    u^* = \max(\min(u^o, u_\tx{max}(\check\omega)), u_\tx{min}(\check\omega)) \label{eq:constrained_control}
\end{align}
    and sub-optimal actuator torques can be obtained from \eqref{eq:constrained_control} and Proposition~\ref{propos:torques}.
\subsection{Design of adaptive LQT}
In this section we propose an LQT to obtain a sub-optimal solution of the optimization problem \eqref{eq:op_LQT}.
    Linear  optimal control methods that work with quadratic performance over control input and regulation/tracking error (LQR/LQT), are well studied and provide an explicit and stable solution.

First, let us choose a proper interior point penalty function.
\begin{lemma} \label{lemma:LQT}
    Let $f^\tx{LQT}_\tx{p}(\check x, x)$ be an interior point penalty function satisfying
\begin{subequations}  \label{eq:conditions_LQT} 
\begin{align}
    &\lim_{x \rightarrow x_\tx{max}} f^\tx{LQT}_\tx{p}(\check x,x) = \infty,\\
    &\lim_{x \rightarrow x_\tx{min}} f^\tx{LQT}_\tx{p}(\check x,x)= \infty,\\
    &\frac{\tx{d}^2 f^\tx{LQT}_\tx{p}}{\tx{d}x^2}(\check x,x) > 0, \quad \forall x \in [x_\tx{min},x_\tx{max}], \label{eq:LQT_penalty_convex}\\
    &\frac{\tx{d}f^\tx{LQT}_\tx{p}}{\tx{d}x}(\check x,\check x) = 0. \label{eq:LQT_penalty_min_at_ref}
\end{align} 
\end{subequations}%
    Then, function
\begin {align}
\begin{split}
   f_\tx{p}^\tx{LQT}(\cdot)&=-q_\tx{p}\bigg( (x_\tx{max}-\check{x})\log\left(\frac{x_\tx{max}-x}{x_\tx{max}-\check{x}}\right) \\
   &+(\check{x}-x_\tx{min})\log\left(\frac{x-x_\tx{min}}{\check{x}-x_\tx{min}}\right) \bigg) \end{split} \label{eq:LQT_penalty}
\end{align}
with $q_\tx{p}> 0$, is such an interior point function.
\end{lemma}
\begin{proof} 
    The proof can be easily verified by replacing $x$ with $\check{x}$, $x_\tx{min}$ or $x_\tx{max}$ into the penalty function. 
\end{proof}

Condition \eqref{eq:LQT_penalty_convex} ensures that $f_\tx{p}^\tx{LQT}$ is strictly convex, as defined in problem \eqref{eq:op_LQT}, while \eqref{eq:LQT_penalty_min_at_ref} ensures that its minimum is obtained at $\check x$. An illustration of this penalty function is provided in Fig.~\ref{fig:LQT_penalty}.
\begin{figure}[tbp]
    \centering
    \includegraphics[width=0.93\linewidth]{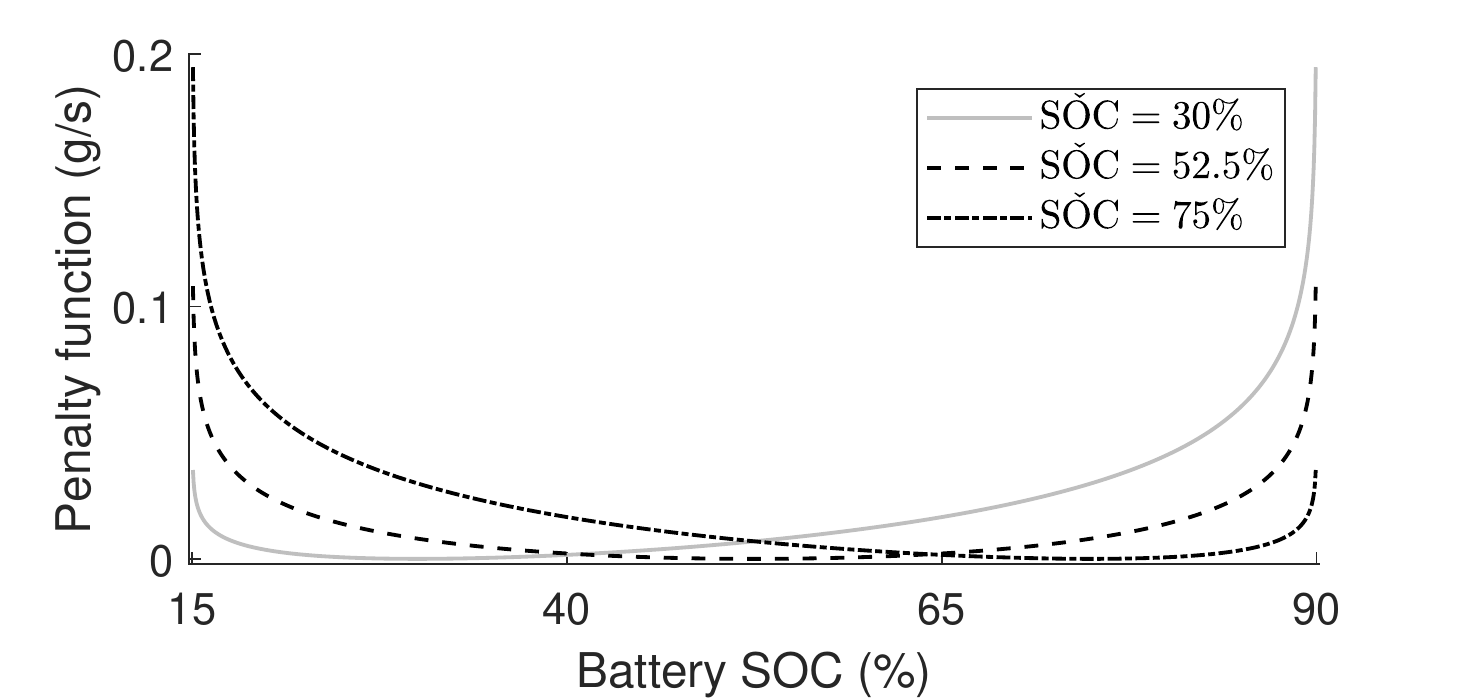}
    \caption{Penalty function used in problem \eqref{eq:op_LQT} for three different values of reference SOC. 
    Battery SOC is allowed to vary within 15-\SI{90}{\%}.}
    \label{fig:LQT_penalty}
\end{figure}

Next, problem \eqref{eq:op_LQT} should be written in a general linear quadratic form. 
It can be observed that it already has linear dynamics as in LQR/LQT, but its objective is not quadratic. Similarly as in ECMS, we use the second order Taylor approximation of fuel consumption in \eqref{eq:taylor2}. 
    It can be rewritten as
\begin {equation}
    \tilde \mu_\tx{fuel} =\tilde{a}_{0} -\frac{\tilde{a}_{1}^2}{4\tilde{a}_{2}} + \tilde{a}_{2} \left(u-u_\tx{max}+\frac{\tilde{a}_{1}}{2\tilde{a}_2}\right)^2
    \label{eq:taylor_LQT2}
\end {equation}
    where the first and second terms in \eqref{eq:taylor_LQT2} are constants that can be removed from the objective as they do not affect the optimal solution.
 
  By using \eqref{eq:taylor_LQT2} and Taylor expansion of the penalty function $f_\tx{p}^\tx{LQT}$ about $\check{x}$, problem \eqref{eq:op_LQT} can be written as   
{\allowdisplaybreaks
\begin{subequations} \label{eq:op_LQT_1}
\begin {align}
    \min_u \, & \frac{1}{2}\sum_{k=0}^{k_f-1} \ \bigl( Q(x(k)-\check{x})^2
    +R(u(k)-\check{u})^2\bigr) \\
    \tx{s.t.: } &  x(k+1) = Ax(k) +Bu(k), \quad x(0)=x_\tx{m}
    \label{eq:LQT_state}
\end{align}
\end{subequations}
}%
    where
\begin {align*}
    &\check{u}=u_\tx{max}-\frac{\tilde{a}_{1}}{2\tilde{a}_2}, \quad A=1, \quad B=-\frac{T_\tx{s}}{E_\tx{Bmax}}, \quad R=2\tilde{a}_2,\\
    &Q=2q_\tx{SOC}+q_\tx{p}\frac{x_\tx{max}-x_\tx{min}} {(x_\tx{max}-\check{x})(\check{x}-x_\tx{min})}.
\end {align*}
     The optimal control of problem \eqref{eq:op_LQT_1} is calculated as
\begin {align}
    u^o(k)=\frac{B \bar P}{R+B^2 \bar P}\left(\check{x}-x(k)\right)
  \label{eq:uss_LQT_1}
\end {align}
    where $\bar P$ is obtained by solving the Riccati equation,
\begin {align}
    \bar P^2-Q \bar P-\frac{QR}{B^2} =0.
    \label{eq:p_1}
\end {align}
    Solving procedure of this problem is presented in Lemma  \ref{lemma:LQTsolving} in Appendix \ref{sec:lemmas}.
    Then, the optimal solution of the constrained problem and optimal actuator torques are calculated using \eqref{eq:constrained_control} and Proposition \ref{propos:torques}.
\begin{remark}
    At every time instance, the infinite horizon LQT in \eqref{eq:op_LQT_1} is solved. Because the reference value of ICE/EM speed, reference SOC and the demanded torque can change at every instant, the weighting coefficients $Q$ and $R$ and therefore the feedback gain may also change.         
\end{remark}
\subsection{Properties of the designed controllers}
\label{sec:comparison_controllers} 
    In this section we discuss some properties of the adaptive ECMS and LQT controllers.
\begin{theorem}
    When the demanded torque satisfies
    \[ M_\tx{Emin}+M_\tx{Meq}\leq \check{M}_\tx{dem} \leq M_\tx{Emax}+M_\tx{Meq}\]
    the optimal ECMS control is nonpositive if $s_\tx{B} \geq s_\tx{B0}$ and is positive if $s_\tx{B} < s_\tx{B0}$, where  
    \begin{align}
    \begin{split}
    s_\tx{B0}=&-\frac{\partial \mu_\tx{fuel}}{\partial u}(u=0)\\
    =&\frac{a_1+2a_2\max(M_\tx{Emin}(\check\omega), \bar M_\tx{dem}(\check\omega) - M_\tx{Meq}(\check\omega))}{\sqrt{d_1^2-4d_2 (d_0 +P_\tx{aux})}}.
     \end{split} \label{eq:sb0}
\end{align}
\label{theorem:ecms_optimalcontrol}
\end{theorem}
     
\begin{proof}
    As described in the proof of Lemma \ref{lemma:s_b},  $s_\tx{B}$ is non-increasing in $u$. For this reason, if $s_\tx{B} \geq s_\tx{B0}$, then $u^o \leq 0$ and if $s_\tx{B}<s_\tx{B0}$, then $u^o>0$. Because $M_\tx{Emin}+M_\tx{Meq}\leq \check{M}_\tx{dem} \leq M_\tx{Emax}+M_\tx{Meq}$, $u^o$ and $u^*$ will be equal if the optimal control is not on the bound, otherwise  according to \eqref{eq:u_min} and \eqref{eq:u_max} they will have the same sign.
\end{proof}

\begin{theorem}
    When demanded torque satisfies
    \[ M_\tx{Emin}+M_\tx{Meq}\leq \check{M}_\tx{dem} \leq M_\tx{Emax}+M_\tx{Meq}\]
    the optimal LQT control is nonpositive if $x(k) \leq \check{x}$ and
    is positive if $x(k)> \check{x}$.
    \label{theorem:uo}
\end{theorem}

\begin{proof}    
    In \eqref{eq:uss_LQT_1}, $B<0$, $\bar P>0$ and $R>0$. Therefore, $u^o(k) \leq 0$ 
    if $x(k) \leq \check{x}$, and $u^o(k) > 0$ 
    if $x(k)>\check{x}$. Because $M_\tx{Emin}+M_\tx{Meq}\leq \check{M}_\tx{dem} \leq M_\tx{Emax}+M_\tx{Meq}$, $u^o$ and $u^*$ will be equal if the optimal control is not on the bound, otherwise  according to \eqref{eq:u_min} and \eqref{eq:u_max} they will have the same sign.  
\end{proof}   

\begin{remark} 
    The term $\check{u}=u_\tx{max}-\tilde{a}_{1}/(2\tilde{a}_2)$ does not directly appear in \eqref{eq:uss_LQT_1} 
    but it has an effect on $R$ and therefore $u^o$. This is because $\bar P>0$ and \eqref{eq:p_1} has one acceptable solution,
\begin {align}
    \bar P=\frac{1}{2}\Bigl(Q+\sqrt{Q^2+4\frac{QR}{B^2}}\Bigr). 
    \label{eq:riccatisolving}
\end {align}
    Using \eqref{eq:uss_LQT_1} and \eqref{eq:riccatisolving} results in
\begin {align}
    u^o(k)=\frac{B\Bigl(Q+\sqrt{Q^2+4\frac{QR}{B^2}}\Bigr)}{2R+B^2\Bigl(Q+\sqrt{Q^2+4\frac{QR}{B^2}}\Bigr)}
    \left(\check{x}-x(k)\right).
\end {align} 
\end{remark}  
\subsection{Robustness of the controllers} \label{sec:robustness}
    In this section we calculate the value of $\epsilon$ such that robustness to state measurement noise is provided.
    
    Let $x_\tx{t}$ 
    denote the true state value and 
    assume that maximum measurement error while evaluating the controller is limited to 
\begin{equation}
    |x_\tx{t}-x_\tx{m}|\leq \beta.  \label{eq:measurement_noise}   
\end{equation}
    To achieve robustness it should be guaranteed that $x_\tx{t}$ and $x_\tx{m}$ are not less than $x_\tx{min}$ or greater than $x_\tx{max}$.
    
    The true state value at instant $k=1$ may violate its lower bound if the following conditions are satisfied
\begin{subequations}
\begin{align}
    &x_\tx{m} > x_\tx{min}+\epsilon,
    \label{eq:measured_state}\\
    &x_\tx{min}<x_\tx{t}(0) < x_\tx{min}+\epsilon,\\
    &u^*(0)>0
\end{align}
\end{subequations}
    and, similarly, may violate its upper bound if
\begin{subequations}
\begin{align}
    &x_\tx{m}< x_\tx{max}-\epsilon,
    \label{eq:measured_state_upper}\\
    &x_\tx{max}-\epsilon<x_\tx{t}(0) < x_\tx{max},\\
    &u^*(0)<0.
\end{align}
\end{subequations}
    To avoid violating the state bounds, we calculate $\epsilon$ such that both true and measured value are within bounds
\begin{align}
   x_\tx{min} \leq x_i \leq x_\tx{max}, \quad i\in\{\tx{t}, \tx{m}\}. 
\end{align}
    According to \eqref{eq:u_min} and \eqref{eq:u_max}, $x_\tx{min}+\epsilon \leq x(1) \leq x_\tx{max}-\epsilon$, then $x_\tx{t}(1)$ can have its minimum value when
\begin{subequations}
\label{eq:xt_min}
\begin {align}
    &x(1)=x_\tx{min}+\epsilon,\\
    &x_\tx{m}-x_\tx{t}(0)= \beta 
\end {align}
\end {subequations}
and it will have its maximum value if
\begin{subequations}
\label{eq:xt_max}
\begin {align}
    &x(1)=x_\tx{max}-\epsilon,\\
    &x_\tx{t}(0)-x_\tx{m}= \beta. 
\end {align}
\end {subequations}
    We have 
\begin {align}
    &x_\tx{t}(1)-x_\tx{t}(0)=x(1)-x_\tx{m}.
    \label{eq:xt}
\end {align}
    Substituting conditions \eqref{eq:xt_min} and \eqref{eq:xt_max} in \eqref{eq:xt} results in
\begin {align}
    &x_\tx{tmin}(1)=x_\tx{min}+\epsilon-\beta
    \label{eq:xt_min_2}\\
    &x_\tx{tmax}(1)=x_\tx{max}-\epsilon+\beta.
    \label{eq:xt_max_2}
\end {align}
    By considering \eqref{eq:measurement_noise}, the measured state when the controller is re-evaluated at next instant is bounded as   
\begin{align}                    
    x_\tx{min}+
    \epsilon-2\beta \leq x_\tx{m} \leq  x_\tx{max}-\epsilon+2\beta.
\end{align}
    To achieve robustness, it should hold
\begin{align}
     &x_\tx{min}+
    \epsilon-2\beta\geq x_\tx{min}, \quad
    x_\tx{max}-\epsilon+2\beta\leq x_\tx{max}
\end{align}
    which can be achieved when
\begin{align}
    \epsilon \geq 2\beta.
    \label{eq:robustness}
\end{align}
\begin{remark}
    If $x_\tx{m} < x_\tx{min}+\epsilon $, then 
\begin{align}
    u_\tx{max}
    =\frac{(x_\tx{m}-x_\tx{min}-\epsilon)E_\tx{Bmax}}{T_\tx{s}} < 0
\end {align}
    and the battery is charged, thus moving away from its lower bound. If $x_\tx{m} > x_\tx{max}-\epsilon $, then 
\begin{align}
    u_\tx{min}
    =\frac{(x_\tx{m}-x_\tx{max}+\epsilon)E_\tx{Bmax}}{T_\tx{s}} > 0
\end {align}
and battery is discharged, thus moving away from its upper bound.
\end{remark}
\section{simulation} \label{sec:simulation}
    In this section, simulation results of adaptive proportional ECMS and LQT are presented. Vehicle and controllers parameters are listed in \tab{tab:parameters} and \tab{tab:controllerparameters}, respectively. The penalty coefficients used in the design of the controllers are chosen by considering a compromise between delivering demanded torque, minimizing fuel consumption and tracking battery SOC. 
    To achieve robustness to state measurement noise, $\epsilon$ is chosen by considering \eqref{eq:robustness}.
    
\begin{table}[tbp]
\centering
\caption{Vehicle parameters.}
{\footnotesize
\begin{tabular}{p{1.5cm} p{1.5cm} p{1.5cm} p{1.3cm}}
    \hline
    Parameter& Value & Parameter& Value\\
   \hline  
 $R_\tx{B}$ & \SI{0.2509}{\Omega} & $E_\tx{Bmax}$ & \SI{36}{MJ} \\
  $U_\tx{oc}$ & \SI{600}{V} & $P_\tx{aux}$ & \SI{2.5}{kW} \\
    \hline
\end{tabular}}%
\label{tab:parameters}
\end{table}

\begin{table}[tbp]
\centering
\caption{Controllers parameters.}
{\footnotesize
\begin{tabular}{p{1.5cm} p{1.5cm} p{1.5cm} p{1.3cm}}
    \hline
    Parameter& Value & Parameter& Value\\
   \hline  
   
    $\tx{SOC}_\tx{min}$ & 15\% & $T_\tx{s}$ & \SI{0.02}{s} \\
    $\tx{SOC}_\tx{max}$ & 90\% & $q_\tx{p}$ & \SI{0.24}{g/s}\\
  $\tx{SOC}(0)$ & 65\%& $q_\tx{SOC}$ & \SI{15}{g/s}\\
  $\beta$ & 0.2\%  & $K_\tx{p1}$& \SI{1.95}{mg/kJ}\\
      $\epsilon$ & 0.5\%&  $K_\tx{p2}$& \SI{37.8}{mg/kJ}\\
    \hline
\end{tabular}}%
\label{tab:controllerparameters}
\end{table}
    The effectiveness of the designed controllers are shown through three examples. 
    The first example is a general case where all reference trajectories can change at any time instant similar to practical applications. The second and the third examples are special cases with the aim of comparing the adaptive ECMS control and LQT method 
    where the speed is constant and the reference trajectory for demanded torque is time-varying positive or negative, in the second and the third example, respectively.
    \subsection{Time-varying reference trajectories for speed and torque}
    For the first example, the reference trajectories of speed and demanded torque between the EM and the gearbox provided by supervisory controller and/or driver are shown in \fig{fig:speed_torque}. For the ECMS control, the reference equivalent factor, $\check{s}_\tx{B}$, is \SI{50.1537}{mg/kJ}.
    The reference trajectory of battery SOC and the SOC obtained by using ECMS and LQT control are depicted in \fig{fig:soc} and  undelivered demanded torque is shown in \fig{fig:undeliveredtorque}. Furthermore, simulation results are summarized in \tab{tab:comparison1}. The results indicate that demanded torque is not delivered  completely at some time instants after \SI{790.5}{s}. This is because the reference SOC is near its lower bound at time instants when the  demanded torque is positive and cannot be delivered by the ICE only.
\begin{figure}[tbp]
    \centering
    \subfigure[Reference trajectories of ICE/EM speed and demanded torque
    .]{
    \includegraphics[width=0.9\linewidth]{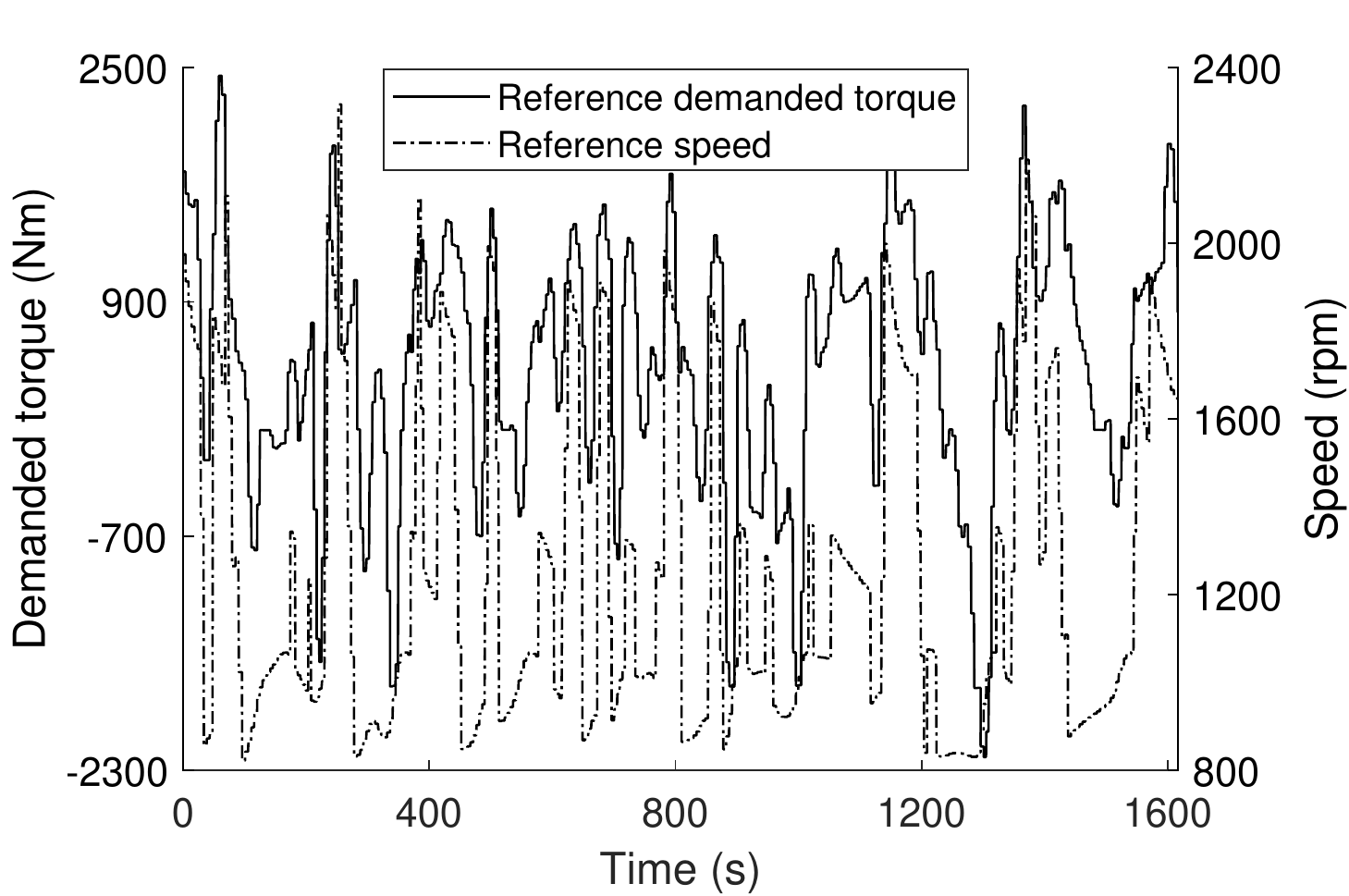}
    \label{fig:speed_torque}
    }
    \subfigure[Battery SOC obtained by ECMS and LQT.]{
    \includegraphics[width=0.9\linewidth]{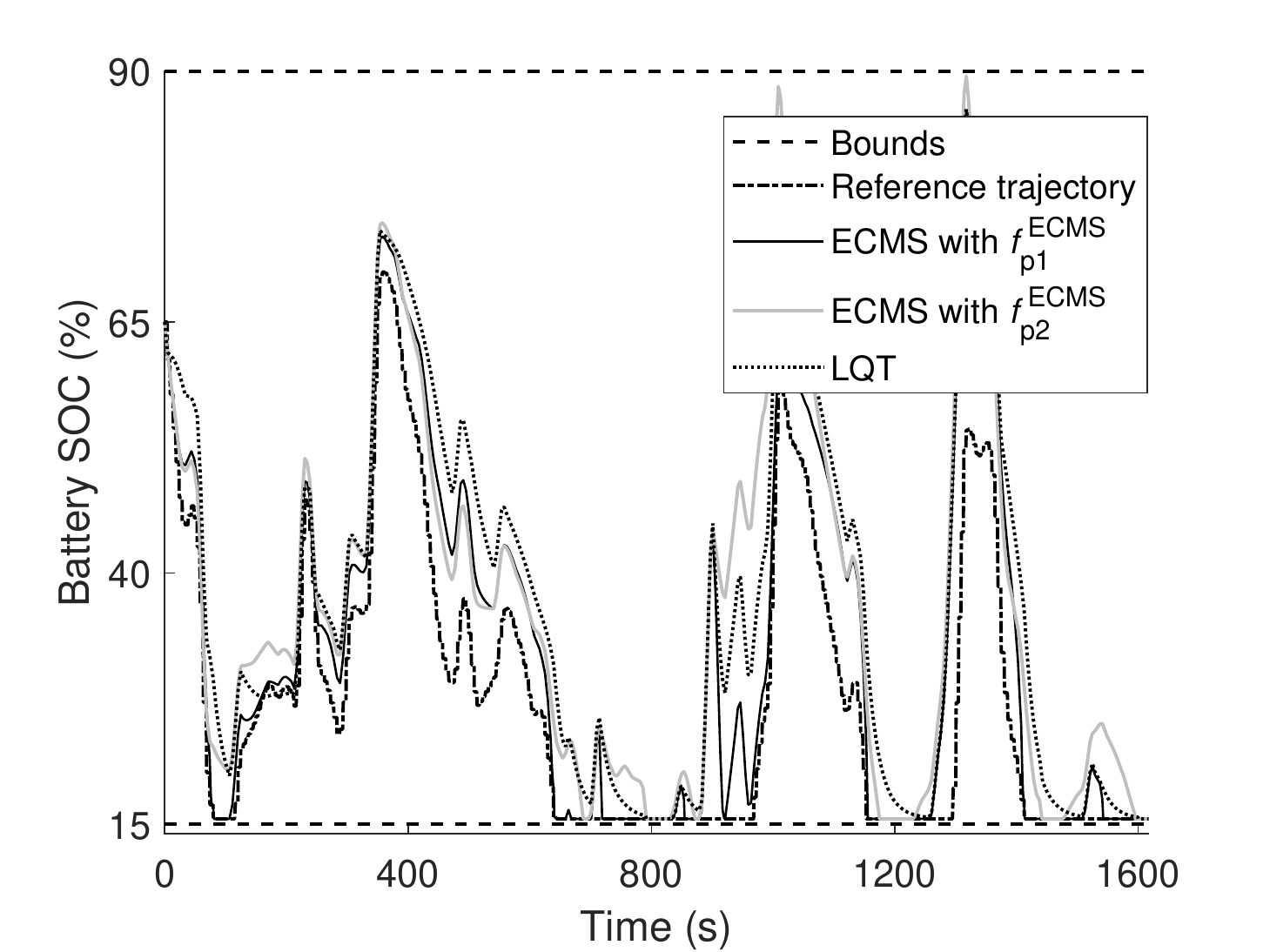}
    \label{fig:soc}
    }
    \subfigure[Undelivered demanded torque by using ECMS and LQT.]{
    \includegraphics[width=0.9\linewidth]{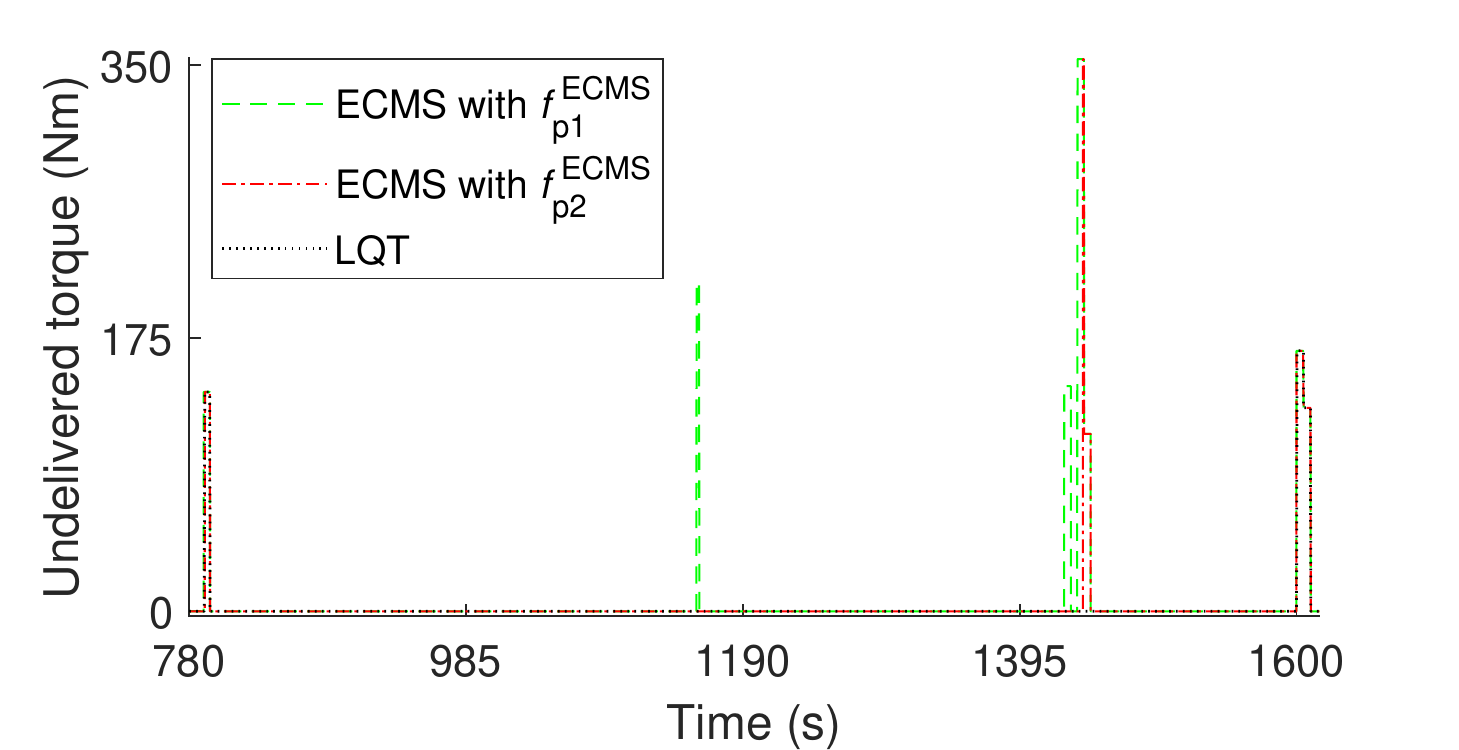}
    \label{fig:undeliveredtorque}
    }
    \caption{SOC trajectories and undelivered torque in the first example where demanded speed and torque can change at any time instant.}
    \label{fig:example1}
\end{figure}
 
 \begin{table}[tbp]
    \centering
    \caption{Comparison of the first example simulation results.}
    {\footnotesize
 \begin{tabular}{p{2.2cm} p{1.3cm} p{2cm}}
    \hline
    Controller& 
    Delivered torque (\%)
    & Average fuel consumption (g/s)\\
    \hline
     ECMS with $f^\tx{ECMS}_\tx{p1}$ & 99.26 & 5.5549\\
     ECMS with $f^\tx{ECMS}_\tx{p2}$ & 99.62 & 5.5409\\
     LQT   & 99.73 & 5.5510 \\
    \hline
 \end{tabular}}%
    \label{tab:comparison1}
\end{table}
    \subsection{Time-varying positive torque reference and constant speed}
    In the second example, the ICE/EM speed is chosen constant at \SI{1974}{rpm} and the other parameters are the same as those given in \tab{tab:controllerparameters}. 
    Furthermore, ECMS control is considered in two cases. In case 1, the reference equivalent factor is the same as that used in the first example. In case 2, it is changed in order to improve the ECMS controllers performance.
    
    The reference trajectory of the demanded torque for this example is depicted in \fig{fig:demandedtorque2}.
    As shown in this figure, the torque is  positive and does not satisfy \eqref{eq:deliverable_torque_infinite_horizon} at the time interval between \SI{235}{s} and \SI{630}{s} and, hence, there is a risk of not being delivered. The reference SOC should be decreased in this interval but it is intentionally and wrongly  chosen to be constant (for example in case of a lack of input data or supervisory control), as shown in \fig{fig:soc2}. The battery SOC obtained by using the ECMS control in two cases and the LQT method is depicted in \fig{fig:soc2} and the controllers are compared in \tab{tab:comparison2}.
    As shown in \fig{fig:soc2}, by using LQT the battery is charged before reference SOC is changed. This is because $x<\check{x}$,  \eqref{eq:deliverable_torque_infinite_horizon} is satisfied and according to Theorem~\ref{theorem:uo} the optimal control is negative. At the time interval when demanded torque jumps to a higher value and does not satisfy \eqref{eq:deliverable_torque_infinite_horizon} (demanded torque is above the dashed line in \fig{fig:demandedtorque2}), The EM torque limit satisfies
    $\hat M_\tx{Mmin}\geq M_\tx{Meq}$, which follows from \eqref{eq:EM_torque_min}. According to the proof of Lemma \ref{lemma:convexity}, $u$ is monotonically increasing with $M_\tx{M}$ and therefore by considering \eqref{eq:u_min}, $u_\tx{min}$ and then $u^*$ are nonnegative and the battery is discharged.
    This is an example where although the unconstrained optimal LQT control tries to track the reference SOC, the control bounds force a constrained solution that causes the battery to be discharged.
    As a consequence, the demanded torque is completely delivered for the entire time horizon. 
    In comparison, by using ECMS control in case 1, the battery is discharged before reference SOC is changed. This is because $s_\tx{B}< s_\tx{B0}$,   \eqref{eq:deliverable_torque_infinite_horizon} is satisfied and according to Theorem~\ref{theorem:ecms_optimalcontrol} the optimal control is positive. 
    The results show that in this case, the ECMS controllers are not able to deliver the demanded torque at some time instances. 
    Then, in case 2 we increase the reference equivalent factor to \SI{52.9651}{mg/kJ}, with the aim of delivering the demanded torque completely. As shown in \fig{fig:soc2}, in this case, the battery is charged before SOC is changed and therefore, the demanded torque can be completely delivered but the average fuel consumption is increased. 

 \begin{figure}[tbp]
    \centering
    \subfigure[Reference trajectory of demanded torque.]{
    \includegraphics[width=0.92\linewidth]{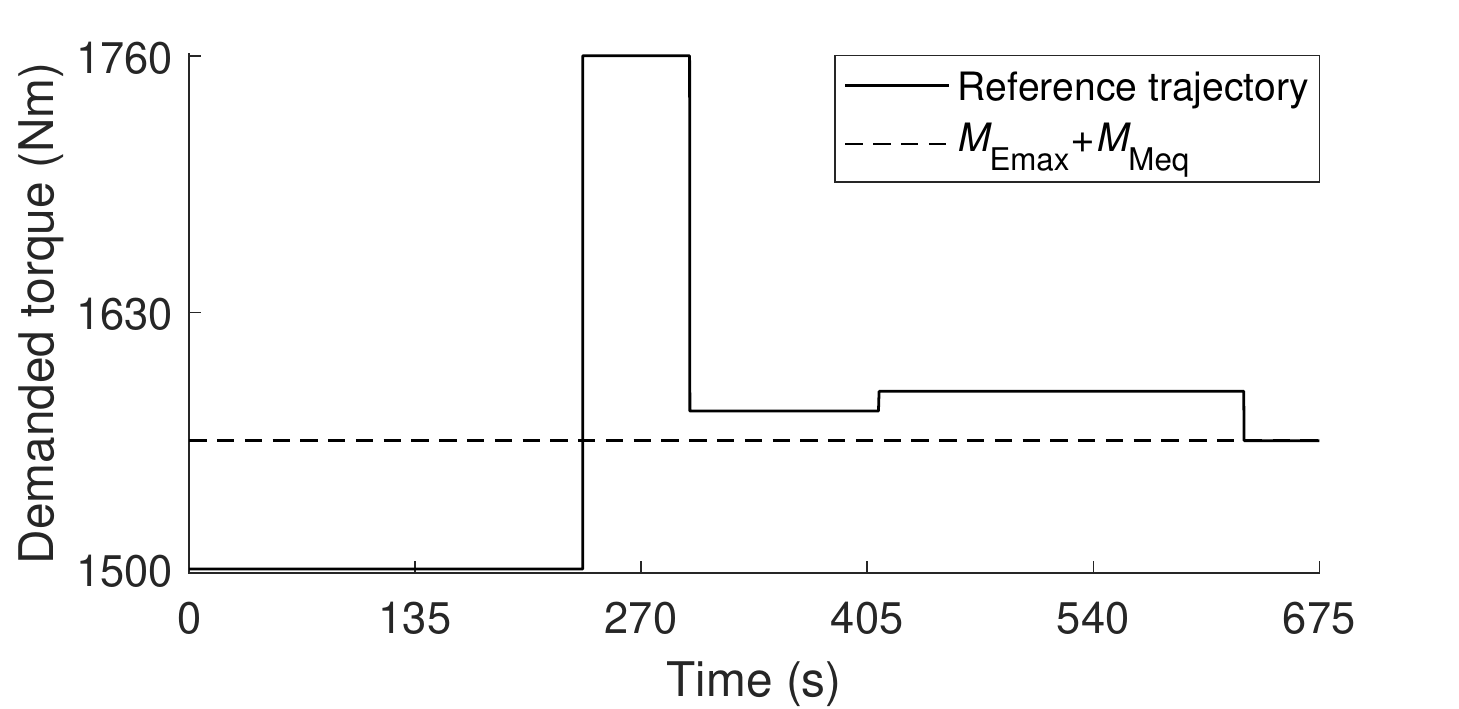}
    \label{fig:demandedtorque2}
    }
%
    \subfigure[Battery SOC obtained by ECMS and LQT.]{
    \includegraphics[width=0.92\linewidth]{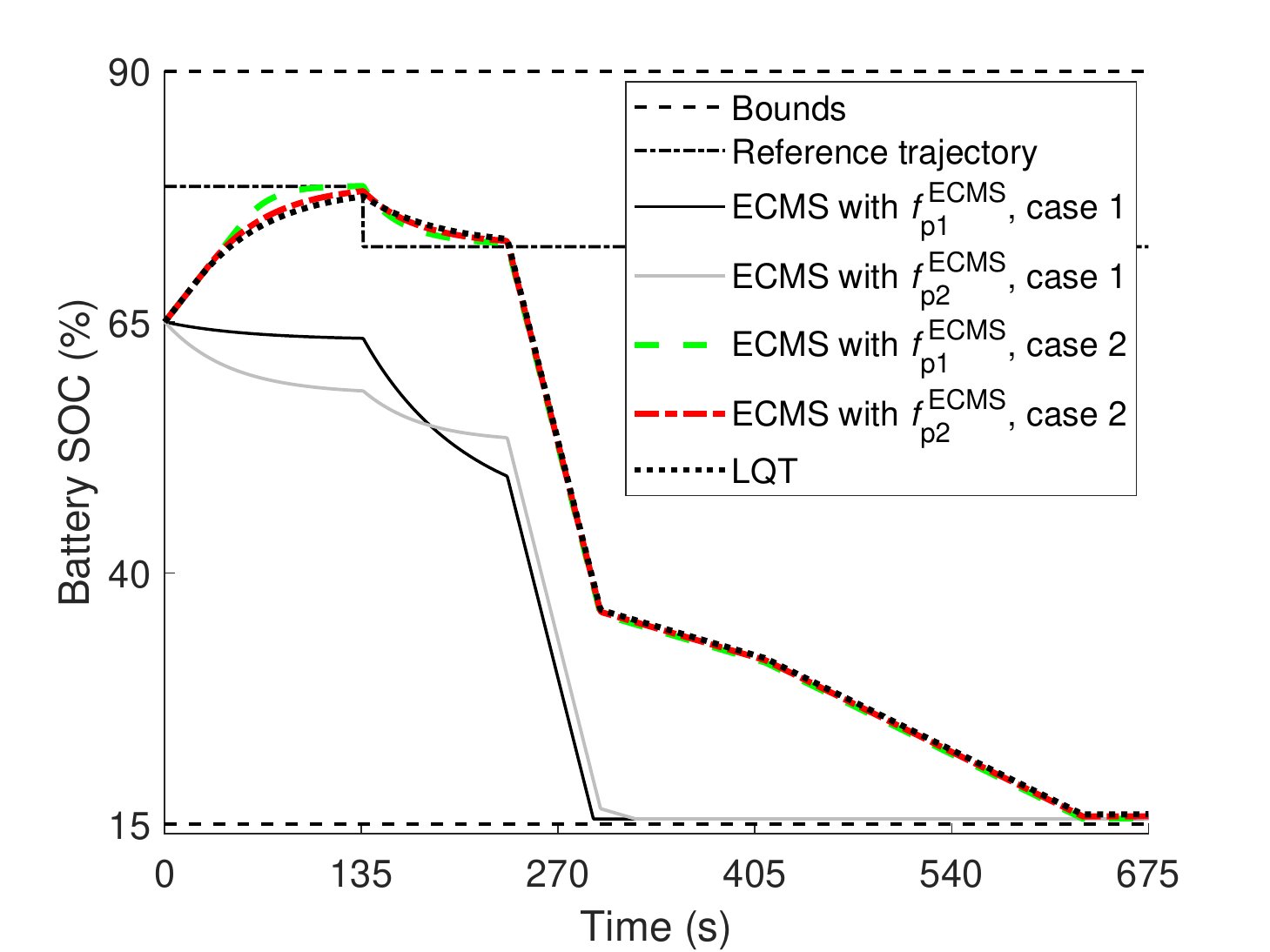}
    \label{fig:soc2}
    }
    \caption{Battery SOC in the second example where demanded torque is positive, piece-wise constant function.}
    \label{fig:example2}
\end{figure}

\begin{table}[tbp]
    \centering
    \caption{Comparison of the second example simulation results.}
{\footnotesize
 \begin{tabular}{p{0.5cm}|p{2.2cm} p{1.3cm} p{2cm}}
    \hline
    Case & Controller&
    Delivered torque (\%)
    & Average fuel consumption (g/s)\\
    \hline
    1 & ECMS with $f^\tx{ECMS}_\tx{p1}$ & 99.24 & 18.7674\\
    & ECMS with $f^\tx{ECMS}_\tx{p2}$ & 99.36 & 18.7894\\
    2 & ECMS with $f^\tx{ECMS}_\tx{p1}$ & 100 & 18.9043\\
    & ECMS with $f^\tx{ECMS}_\tx{p2}$ & 100 & 18.9054\\
    & LQT   & 100 & 18.9064\\
    \hline
\end{tabular}}%
    \label{tab:comparison2}
\end{table}
    
    \begin{figure}[tbp]
    \centering
    \subfigure[Reference trajectory of demanded torque.]{
    \includegraphics[width=0.92\linewidth]{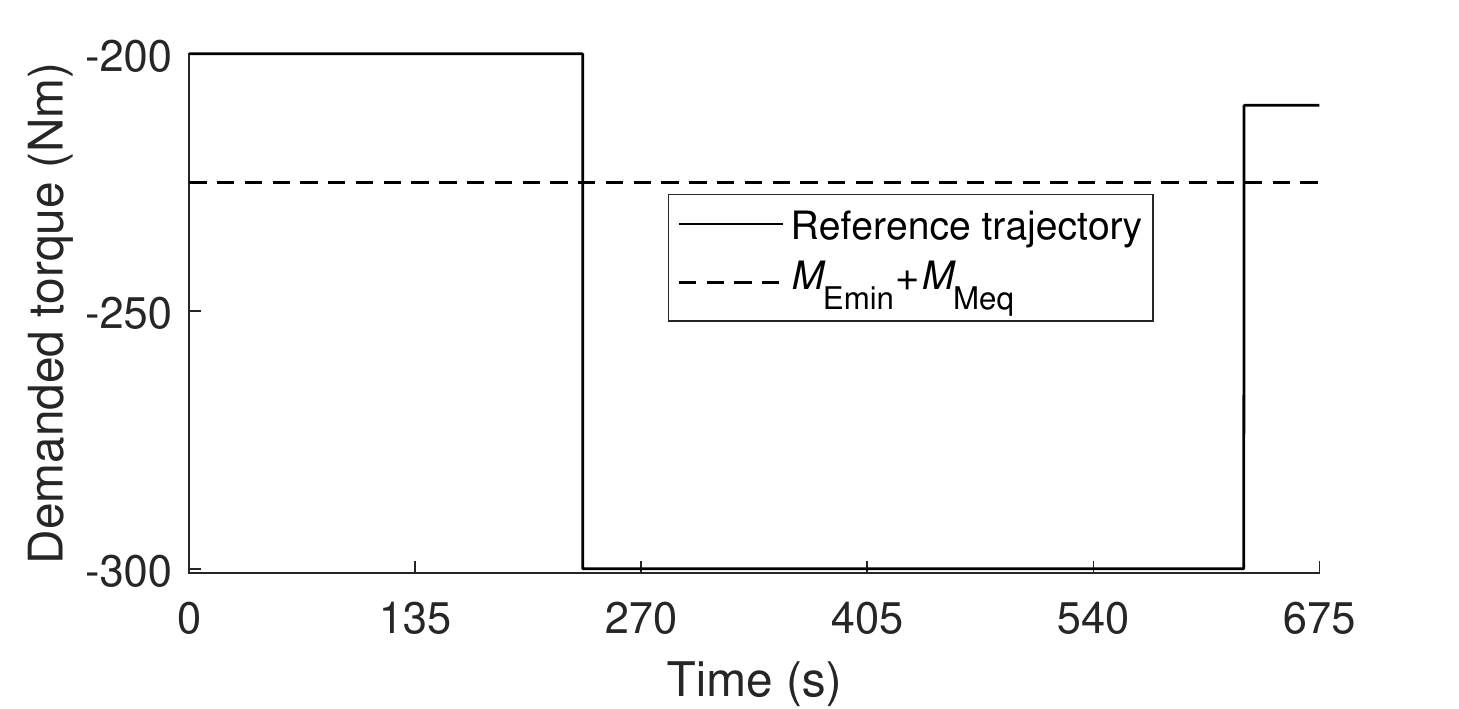}
    \label{fig:demandedtorque3}
    }
%
    \subfigure[Battery SOC obtained by ECMS and LQT.]{
    \includegraphics[width=0.92\linewidth]{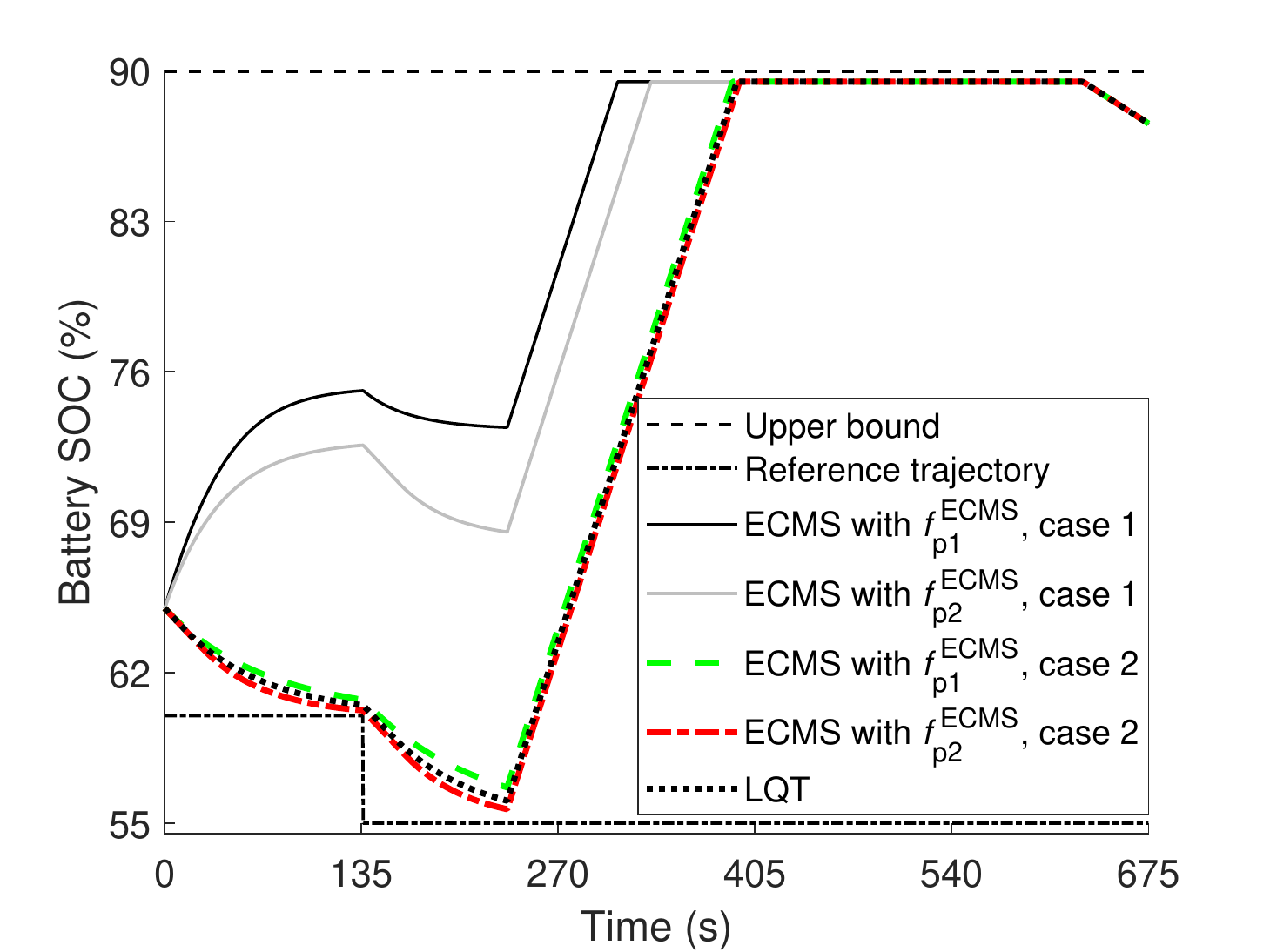}
    \label{fig:soc3}
    }
    \caption{Battery SOC in the third example where demanded torque is negative, piece-wise constant function.}
    \label{fig:example3}
\end{figure}
    \subsection{Time-varying negative torque reference and constant speed}
    In the third example, the parameters are the same as those in the second example.
    This example studies a negative demanded torque, with a reference trajectory as depicted in \fig{fig:demandedtorque3}. As shown in the figure, the demanded torque has first a low negative magnitude and then, between \SI{235}{s} and \SI{630}{s}, its negative magnitude increases below $M_\tx{Emin}+M_\tx{Meq}$. 
    The reference SOC should be increased in this interval but it is intentionally and wrongly chosen to be constant (e.g., for the case when supervisory control or input data is not available), as shown in \fig{fig:soc3}.
    ECMS control is investigated in two cases similar to the second example. 
    The battery SOC obtained by using the ECMS in two cases of the equivalent factor and the LQT control are depicted in \fig{fig:soc3} and simulation results are given in \tab{tab:comparison2}. 
    \begin{table}[t]
    \centering
    \caption{Comparison of the third example simulation results.}
    {\footnotesize
 \begin{tabular}{p{0.5cm}|p{2.2cm} p{1.3cm} p{2cm}}
    \hline
    Case & Controller&
    Delivered torque (\%)
    & Average fuel consumption (g/s)\\
    \hline
    1 & ECMS with $f^\tx{ECMS}_\tx{p1}$ & 100 & 0.1301\\
    & ECMS with $f^\tx{ECMS}_\tx{p2}$ & 100 & 0.1050\\
    2 & ECMS with $f^\tx{ECMS}_\tx{p1}$ & 100 & 0.0438\\
    & ECMS with $f^\tx{ECMS}_\tx{p2}$ & 100 & 0.0387\\
    & LQT   & 100 & 0.0407\\
    \hline
\end{tabular}}%
    \label{tab:comparison3}
\end{table}
    As shown in \fig{fig:soc3}, by using LQT, battery is discharged before reference SOC is changed, because $x>\check{x}$, $\check{M}_\tx{dem}>M_\tx{Emin}+M_\tx{Meq}$ and according to Theorem~\ref{theorem:uo}, optimal control is positive. In other words, the demanded torque magnitude is not negative enough to charge the battery and it is, instead, absorbed by the auxiliaries and by the ICE, to overcome its friction losses, and thus decreasing fuel consumption. Recall that we investigate scenarios where the engine is kept on, so in this case fuel can be saved by having the vehicle running the engine, instead of the other way around. At the time interval when $\check{M}_\tx{dem} \leq M_\tx{Emin}+M_\tx{Meq}$, the EM torque limit satisfies $\hat M_\tx{Mmax}\leq M_\tx{Meq}$, which follows from \eqref{eq:EM_torque_max}. This is because according to the proof of Lemma \ref{lemma:convexity}, $u$ is monotonically increasing with $M_\tx{M}$, and by considering \eqref{eq:u_max}, $u_\tx{max}$ and then $u^*$ are nonpositive and the battery is charged. This is an example where although the unconstrained optimum from LQT tries to track the reference SOC and discharge the battery, the control bounds force a constrained solution that instead charges the battery. In comparison, by using ECMS control in case 1, the battery is charged before reference SOC is changed. This is because $s_\tx{B}> s_\tx{B0}$, $\check{M}_\tx{dem}>M_\tx{Emin}+M_\tx{Meq}$ and according to Theorem~\ref{theorem:ecms_optimalcontrol}, the optimal control is negative. 
    As a consequence, the average fuel consumption is greater than that of LQT.
    In order to decrease the average fuel consumption, we decrease the reference equivalent factor to \SI{48.0375}{mg/kJ} in case 2. As shown in \fig{fig:soc3}, in this case,  the battery is discharged before SOC is changed and therefore, the average fuel consumption is decreased.
    
    Results of the second and the third examples can be summarized by the following remark.
\begin{remark}
    Because delivering the demanded torque in an HEV is the first objective of this paper, when battery SOC is near its lower bound and ${\check{M}_\tx{dem} > M_\tx{Emax}+M_\tx{Meq}>0}$, the battery is discharged until ${x=x_\tx{min}+\epsilon}$ and if battery SOC is near the upper bound and $ {\check{M}_\tx{dem} < M_\tx{Emin}+M_\tx{Meq}<0}$, to minimize usage of service brakes, the battery is charged until ${x=x_\tx{max}-\epsilon}$ regardless of the reference SOC value.
\end{remark}
\section{Conclusions} \label{sec:conclusion}
    In this paper, to deliver the torque as close as possible to demanded torque, a function is proposed for calculating the maximum deliverable torque. Then, the number of control inputs in the optimization problem is reduced to one and therefore the computational complexity  for optimal control calculation is reduced, while usage of service brakes in the case of negative demanded torque is minimized. After that,
    adaptive proportional ECMS control and an LQT are designed at every time instance to optimize power-split decisions in a parallel HEV by minimizing fuel consumption subject to constraints on actuators and battery SOC. The SOC 
    limitations are modeled by using tangent or logarithm interior point functions.
    Then, the convexity of the resulting optimization problem is proved and it is analytically solved by using the second order approximation of the objective function.
    
    The simulation results obtained using the ECMS control with two different penalty functions and the LQT controller show that their performance is a trade-off between delivering demanded torque, minimizing fuel consumption and tracking reference SOC.
    Moreover, they indicate that the reference equivalent factor has a very important role in the performance of the ECMS controller. If it cannot be correctly provided by a supervisory controller, setting it to a proper value is not straightforward. On the other hand, the LQT controller requires only setting a reference battery SOC, which can be either provided by a supervisory controller, or set manually. 
    
    In this paper a simple measurement noise model is considered when designing the controllers. Future research may focus on considering more detailed uncertainty and noise models, as well as more detailed HEV powertrain models with additional dynamics, e.g. thermal states of the battery, EM, or the ICE.
    
\begin{appendices}
\section{Proof of lemmas}\label{sec:lemmas}

\begin{lemma} \label{lemma:em_torque}
    Equation \eqref{eq:electricpower} has a single solution 
\begin {align}
    M_\tx{M}&= \frac{-d_{1} +\sqrt{ d_{1}^2-4d_{2} (d_{0}-P_\tx{Mel}) } }{2d_{2} }. \label{eq:EM_torque}
\end {align}
    Thereby, the EM torque is implicitly constrained to
\begin{align} \label{eq:pmelmin}
    M_\tx{M} \geq - \frac{d_1}{2d_2}
\end{align}
    where minimum electrical power is achieved.
\end{lemma}

\begin{proof}
    Because $d_2 \geq 0$, $P_\tx{Mel}$ is a convex function of $M_\tx{M}$ with minimum obtained at $\partial P_\tx{Mel}/\partial M_\tx{M}=0$. Hence, minimum electrical power is $d_0 - d_1^2/(4d_2)$ obtained at torque $-d_1/(2d_2)$.
    One of the roots of \eqref{eq:electricpower} is
\begin {align} \label{eq:infeas_root_em_torque}
    M_\tx{M}=\frac{-d_{1} -\sqrt{ d_{1}^2-4d_{2} (d_{0}-P_\tx{Mel}) } }{2d_{2} }
\end {align}
    for which it holds
\begin {align}
    &\frac{\partial M_\tx{M}}{\partial P_\tx{Mel}}=-\frac{1}{\sqrt{d_{1}^2-4d_{2} (d_{0}-P_\tx{Mel})}}<0.
\end{align}
    This is a clear contradiction with a physical EM component where electrical power is monotonically increasing with torque. Hence, the root \eqref{eq:infeas_root_em_torque} cannot be a solution to \eqref{eq:electricpower}. Furthermore, it is clear that \eqref{eq:infeas_root_em_torque} can only provide negative real values for $M_\tx{M}$, thus not allowing motoring mode.
    EM electrical power is increasing with torque where $M_\tx{M} \geq -d_1/(2d_2)$. Equation \eqref{eq:EM_torque} implies this constraint on torque too. 
\end{proof}

\begin{lemma} \label{lemma:batterypower}
    Equation \eqref{eq:elbalance} has a single solution,
\begin{align} \label{eq:batterypower}
     {P}_\tx{B}=&\frac{U_\tx{oc}^2 - U_\tx{oc}\sqrt{U_\tx{oc}^2 - 4R_\tx{B} P_\tx{Bel}} }{2R_\tx{B}}.
\end{align}
    Thereby, chemical battery power is implicitly constrained, 
\begin{align}
    P_\tx{B} \leq \frac{U_\tx{oc}^2}{2 R_\tx{B}}. 
\end{align}
\end{lemma}

\begin{proof}
    It can be seen from \eqref{eq:elbalance} that $P_\tx{Bel}$ is a concave function of $P_\tx{B}$. Its maximum is obtained at ${\partial P_\tx{Bel}/\partial P_\tx{B}=0}$ and battery power $U_\tx{oc}^2/(2 R_\tx{B})$.\\ 
    One of the roots of \eqref{eq:elbalance} is 
\begin{align} \label{eq:infeas_sol_pwr_balance}
    P_\tx{B} =&\frac{U_\tx{oc}^2 + U_\tx{oc}\sqrt{U_\tx{oc}^2 - 4R_\tx{B} P_\tx{Bel}} }{2R_\tx{B}}
\end{align}
    for which it holds
\begin {align} \label{eq:reason2}
    &\frac{\partial P_\tx{B}}{\partial P_\tx{Bel}}=-\frac{U_\tx{oc}}{\sqrt{U_\tx{oc}^2 - 4R_\tx{B} P_\tx{Bel}}} < 0.
\end{align}
    This is a clear contradiction with a physical battery component where chemical battery power is monotonically increasing with $P_\tx{Bel}$. Hence, the root \eqref{eq:infeas_sol_pwr_balance} cannot be a solution to \eqref{eq:elbalance}. Furthermore, it is clear that \eqref{eq:infeas_sol_pwr_balance} can only provide positive real values for $P_\tx{B}$, thus not allowing battery charging.\\
    Chemical battery power is increasing with $P_\tx{Bel}$ where $P_\tx{B} \leq U_\tx{oc}^2/(2 R_\tx{B})$. 
    This constraint on $P_\tx{B}$ is implied by \eqref{eq:batterypower} too. 
\end{proof}

\begin{lemma}
    The optimal control of problem \eqref{eq:op_LQT_1} is 
\begin {align}
    u^o(k)=\frac{B \bar P}{R+B^2 \bar P}\left(\check{x}-x(k)\right)
    \label{eq:uss_LQT}
\end {align}
    where 
    $\bar P$ is calculated by solving the Riccati equation,
\begin {align}    
    \bar P^2-Q \bar P-\frac{QR}{B^2} =0.
\label{eq:p}
\end {align}
\label{lemma:LQTsolving}
\end{lemma}
\begin{proof}
The objective function of problem \eqref{eq:op_LQT_1} can be generally written as 
\begin{align}
\begin{split}
  \frac{1}{2}&\sum_{k=0}^{k_f-1} \ \left((x(k)-\check{x}(k))^T Q(k)(x(k)-\check{x}(k)\right.)\\
    &\left.+((u(k)-\check{u}(k))^T R(k)(u(k)-\check{u}(k)) \right) 
\end{split}
\end{align}
    For this optimization problem, the Hamiltonian is  
\begin {align}
\begin {split}
    H&=\frac{1}{2} \bigl(x(k)-\check{x}(k)\bigr)^T Q(k)\bigl(x(k)-\check{x}(k)\bigr) \\
    &+\frac{1}{2}\bigl((u(k)-\check{u}(k)\bigr)^T 
    R(k)\bigl(u(k)-\check{u}(k)\bigr) \\
    &+\lambda^T(k+1)\bigl(Ax(k)+Bu(k)\bigr).
\end {split}    
\end {align}
    Optimal control is achieved when 	 
\begin {align}
\begin {split}
    \frac{\partial H}{\partial u(k)}=0 \; \Rightarrow \; u(k)&=-R^{-1}(k)B^T\lambda(k+1)+\check{u}(k).\label{eq:control}
\end {split}
\end {align}
    The state equation is obtained as 
\begin {align}
\begin {split}
    \frac{\partial H}{\partial \lambda(k+1)}=x&(k+1) \quad \Rightarrow \quad 
    x(k+1)=Ax(k)\\
    &-BR^{-1}(k)B^T\lambda(k+1)
    +B\check{u}(k)\label{eq:state}
\end {split}    
\end {align} 
    and the costate equation is calculated as 
\begin {align}
\begin {split}
    \frac{\partial H}{\partial x(k)}=\lambda(k) \quad \Rightarrow \quad &\lambda(k)=A^T(k)\lambda(k+1)\\
    &+Q(k)x(k)- Q(k)\check{x}(k).\label{eq:costate}
\end {split}    
\end {align} 
    The costate is defined as
\begin {equation}
\lambda(k)=P(k)x(k)-G(k).\label{eq:lambda}
\end {equation}
    For the sake of simplicity, we define ${N=BR^{-1}(k)B^{T}}$.
    Substituting \eqref{eq:lambda} in \eqref{eq:state} results in 
\begin {align}
\begin {split}
    x(k+1)&=\left(\tx{I}+NP(k+1)\right)^{-1}\left(Ax(k)
    \right.\\
    &\left.+NG(k+1)+B\check{u}(k)\right)\label{eq:x1}
\end {split}    
\end {align}
    and by using \eqref{eq:lambda} and \eqref{eq:x1}, costate equation \eqref{eq:costate} can be written as
\begin {align}
    &\left(-P(k)+A^TP(k+1)(\tx{I}+NP(k+1))^{-1}A+Q(k) \right)\nonumber\\
    &x(k)+A^T\left(P(k+1)(\tx{I}+NP(k+1))^{-1}N-\tx{I}\right)G(k+1)\nonumber \\
    &+G(k)+A^TP(k+1)(\tx{I}+NP(k+1))^{-1}B\check{u}(k)\nonumber\\
    &-Q(k)\check{x}(k)=0.
    \label{eq:ri}
\end {align}  
    Because \eqref{eq:ri} must hold for all optimum values of $x(k)$, then
\begin {align}
\begin {split}
    &A^TP(k+1)\left(\tx{I}+NP(k+1)\right)^{-1}A\\
    &-P(k)+Q(k) =0 
\end {split} \label{eq:riccati}\\   
    &A^T\left(P(k+1)(\tx{I}+NP(k+1))^{-1}N-\tx{I}\right)G(k+1)\nonumber\\
    &+G(k)+A^TP(k+1)(\tx{I}+NP(k+1))^{-1}B\check{u}(k)\nonumber\\
    &-Q(k)\check{x}(k)=0. \label{eq:g}
\end {align}
    By solving the matrix difference Riccati equation \eqref{eq:riccati}, $P(k)$ is obtained and $G(k+1)$ is computed from \eqref{eq:g}.
    Then, by substituting \eqref{eq:lambda} in   \eqref{eq:control} the optimal control is calculated as 
\begin {equation}
    u^o(k)=-L_1(k)x(k)+L_2(k)G(k+1)+L_3(k)\check{u}(k)
    \label{eq:input_LQT}
\end {equation}
    where
\begin {align*}
    L_1(k)=&(R(k)+B^TP(k+1)B)^{-1}
    B^TP(k+1)A 
    \\
    L_2(k)=&(R(k)+B^TP(k+1)B)^{-1}B^T
    \\
    L_3(k)=&(R(k)+B^TP(k+1)B)^{-1}R(k).
\end {align*}
    As $k_f \rightarrow \infty$, the objective function in problem \eqref{eq:op_LQT_1} is not finite, an average cost is considered \cite{willemsa04},
\begin{align}
\begin{split}
      \lim_{k_f\to\infty} \frac{1}{2k_f}\sum_{k=0}^{k_f-1} \ &\bigl( Q(x(k)-\check{x})^2+R(u(k)-\check{u})^2\bigr).
    \label{eq:obj_LQT_ss}
\end{split}
\end{align} 
    If the state space equation in \eqref{eq:LQT_state} is in steady state, then equilibrium control, ${u_\tx{eq}=0}$ and objective \eqref{eq:obj_LQT_ss} has its minimum value when equilibrium optimum state, ${x_\tx{eq}=\check{x}}$. 
%
    Furthermore, because $k_f \rightarrow \infty$, $P$ in \eqref{eq:riccati} tends to its steady state value $\bar P$ \cite{naidu02}. By replacing $P(k+1)$ and $P(k)$ with $\bar P$, the Riccati equation is obtained as \eqref{eq:p}. 
     Using \eqref{eq:lambda} and substituting $x_\tx{eq}$ in \eqref{eq:costate} gives $G(k)=G(k+1)$. By replacing $G(k+1)$ and $G(k)$ with its steady state value $\bar G$ in \eqref{eq:g} and using \eqref{eq:p}, 
\begin {align}
    \bar G=P\check{x}-\frac{R}{B}\check{u}. 
\end {align}
    It can be obtained by using definition \eqref{eq:lambda} and replacing $u$ and $x$ with $\check{u}$ and $\check{x}$ into \eqref{eq:control} too.
    By substituting $A$, $\bar P$ and $\bar G$ in \eqref{eq:input_LQT}, $u^o$ is computed as in \eqref{eq:uss_LQT}.
\end{proof}
\end{appendices}
\bibliographystyle{IEEEtran}
\bibliography{IEEEabrv,paper1}

\end{document}